\documentclass{IEEEtran}
\DeclareUnicodeCharacter{0301}{\'{e}}
\usepackage{comment}
\usepackage{algorithm,algorithmicx,algpseudocode,amsmath,amssymb,amsthm}
\usepackage{bm,color,dsfont,pifont,subfigure,hyperref}
\usepackage{makecell}
\usepackage{multirow}
\usepackage{scalerel}
\usepackage{url}
\usepackage{mathtools}
\usepackage{authblk}
\usepackage{pifont}
\usepackage{fontawesome}
\usepackage{wasysym}
\usepackage{cite}
\usepackage{subfigure}

\usepackage{array}
\usepackage{booktabs}

\makeatletter
\renewcommand{\maketag@@@}[1]{\hbox{\m@th\normalsize\normalfont#1}}%
\makeatother

\newtheorem{fact}{Fact}

\newtheorem{assumption}{Assumption}
\newtheorem{definition}{Definition}

\newtheorem{lemma}{Lemma}
\newtheorem{theorem}{Theorem} 
\newtheorem{remark}{Remark}
\newtheorem{proposition}{Proposition}

\allowdisplaybreaks
\newcommand\numberthis{\addtocounter{equation}{1}\tag{\theequation}}

\newcommand{\ip}[2]{\left\langle#1,#2\right\rangle}
\newcommand{\norm}[1]{\left\lVert#1\right\rVert}
\renewcommand{\(}{\left(}
\renewcommand{\)}{\right)}
\newcommand{\bfm}[1]{\bm{#1}}

\def\la{\left\langle}
\def\ra{\right\rangle}

\def\ln{\left\|}
\def\rn{\right\|}
\def\lsb{\left[}
\def\rsb{\right]}

\def\A{\mathcal{A}}
\def\B{\mathcal{B}}
\def\C{\mathbb{C}}

\def\P{\mathcal{P}}
\def\T{\mathcal{T}}
\def\H{\mathcal{H}}
\def\D{\mathcal{D}}
\def\G{\mathcal{G}}
\def\R{\mathbb{R}}
\def\E{\mathbb{E}}

\def\I{\mathcal{I}}

   \def\mA{\bfm A}  
  \def\mB{\bfm B} 
  \def\mC{\bfm C}  
  \def\mD{\bfm D}  
  \def\mE{\bfm E}  
 \def\mF{\bfm F}  
   \def\mG{\bfm G} 
   
  \def\mI{\bfm I}

   \def\mL{\bfm L}
  \def\mM{\bfm M}

  \def\mQ{\bfm Q}  
   \def\mR{\bfm R}

   \def\mU{\bfm U} 
 \def\mV{\bfm V}  
    
   \def\mX{\bfm X}  
   \def\mY{\bfm Y}  
   \def\mZ{\bfm Z}

\newcommand{\va}{\bm{a}}
\newcommand{\vb}{\bm{b}}

\newcommand{\ve}{\bm{e}}
\newcommand{\vg}{\bm{g}}
\newcommand{\vh}{\bm{h}}

\newcommand{\vx}{\bm{x}}

\newcommand{\vw}{\bm{w}}
\newcommand{\vy}{\bm{y}}
\newcommand{\vz}{\bm{z}}

\newcommand{\vX}{\bm{X}}

\newcommand{\vZ}{\bm{Z}}

\newcommand{\vD}{\bm{D}}
\newcommand{\vE}{\bm{E}}
\newcommand{\vF}{\bm{F}}

\newcommand{\vR}{\bm{R}}
\newcommand{\vM}{\bm{M}}

\newcommand{\vU}{\bm{U}}
\newcommand{\vV}{\bm{V}}

\newcommand{\vP}{\bm{P}}
\newcommand{\vQ}{\bm{Q}}

\def\vSS{\bm{\Sigma}}
\newcommand \vDelta{\bm{\Delta}}

\newcommand \vDeltal{\bm{\Delta}_L}
\newcommand \vDeltar{\bm{\Delta}_R}

\def\bzero{\bm{0}}
\def\bPhi {\bm{\Phi}}
\def\bPsi {\bm {\Psi}}

\DeclareMathOperator{\rank}{rank}
\DeclareMathOperator{\diag}{diag}
\DeclareMathOperator{\Real}{{\normalfont Re}}

\newcommand{\dist}[2]{{\normalfont\mbox{dist}}(#1,#2)}
\newcommand{\distsq}[2]{{\normalfont\mbox{dist}^2}(#1,#2)}

\makeatletter
\renewcommand*{\@fnsymbol}[1]{\ensuremath{\ifcase#1\or \dagger\or \ddagger\or \mathsection\or
    *\or \mathparagraph\or \|\or **\or \dagger\dagger
    \or \ddagger\ddagger \else\@ctrerr\fi}}
\makeatother
\title{Simpler Gradient Methods for Blind Super-Resolution with Lower Iteration Complexity}

\author{Jinsheng Li, Wei Cui, and Xu Zhang
\thanks{This work was supported in part by the National Natural Science Foundation of China under Grant 62025103 and the Postdoctoral Fellowship Program of CPSF under Grant GZC20232038. \textit{(Corresponding author: Xu Zhang.)}}
\thanks{
J. Li and W. Cui are with the School of Information and Electronics, Beijing Institute of
Technology, Beijing 100081, China (e-mail: jinshengli@bit.edu.cn; cuiwei@bit.edu.cn). 
X.~Zhang is with the School of Artificial Intelligence, Xidian University, Xi'an 710126, China (e-mail: zhang.xu@xidian.edu.cn). 
}}

\begin{document}


\maketitle
\begin{abstract}
We study the problem of blind super-resolution, which can be formulated as a low-rank matrix recovery problem via vectorized Hankel lift (VHL). The previous gradient descent method based on VHL named PGD-VHL relies on additional regularization such as the projection and balancing penalty, exhibiting a suboptimal iteration complexity. In this paper, we propose a simpler unconstrained optimization problem without the above two types of regularization and develop two new and provable gradient methods named VGD-VHL and ScalGD-VHL. A novel and sharp analysis is provided for the theoretical guarantees of our algorithms, which demonstrates that our methods offer lower iteration complexity than PGD-VHL. In addition, ScalGD-VHL has the lowest iteration complexity while being independent of the condition number. {Furthermore, our novel analysis reveals that the blind super-resolution problem is less incoherence-demanding, thereby eliminating the necessity for incoherent projections to achieve linear convergence.} Empirical results illustrate that our methods exhibit superior computational efficiency while achieving comparable recovery performance to prior arts.
\end{abstract}
\section{Introduction}
Super-resolution refers to enhancing high-resolution details from coarse-scale observations, which is an essential problem in various applications including single-molecule imaging \cite{Quirin2011}, radar target detection\cite{Bajwa2011}, and astronomy \cite{Puschmann2005}. 
In specific contexts like radar and communication systems \cite{Candes2013a,Zheng2017,Heckel2016}, super-resolution involves estimating 
the locations of some point sources based on known point spread functions (PSFs). However, the PSFs might not be available due to imperfect data acquisition, unknown calibrations, or varying temporal and spatial characteristics, which are inherent in scenarios such as blind deconvolution of seismic data \cite{Margrave2011}, blind channel identification \cite{Luo2006}, and 2D microscopy imaging \cite{Rust2006}. In this paper, we focus on this blind super-resolution scenario and aim to resolve the locations of point sources from their convolution with unknown PSFs.

 
Under the assumption that the PSFs live in a known low-dimensional subspace,  several studies \cite{Chi2016, Yang2016,Chen2022,Suliman2022,Mao2022} employed the lifting trick \cite{Ahmed2014} and formulated blind super-resolution of point sources as a matrix recovery problem. Chen et al. exploited the low-dimensional structure of the target matrix via a vectorized Hankel lift (VHL) framework and cast blind super-resolution as a nuclear norm minimization problem \cite{Chen2022}, which can be time-consuming for large-scale scenarios. To address the computational challenge, Mao and Chen proposed an advanced projected gradient descent method via VHL named PGD-VHL based on low-rank matrix factorization in \cite{Mao2022}. 


In particular, PGD-VHL \cite{Mao2022}  formulated the following constrained optimization problem by introducing two types of regularization, that is, projection and the balancing penalty:
\begin{align*}
\min_{\mF\in\mathcal{M} }~ f(\mF)+ \frac{1}{16}\|{\mL^H\mL - \mR^H\mR}\|_F^2,
\end{align*}
where $f(\cdot)$ is the loss function constructed for blind super-resolution via VHL, $\mF=\begin{bmatrix}
 	\mL^H &\mR^H
 \end{bmatrix}^H$ for $\mL\in\C^{sn_1\times r}$  and $\mR\in\C^{n_2\times r}$, $\mathcal{M}$ denotes the set 
\begin{small}
\begin{align*}
 	\mathcal{M} = \bigg\{ \begin{bmatrix}
 		\mL\\
 		\mR\\
 	\end{bmatrix}~:~ \max_{0\leq j\leq n_1-1} \|{\mL_j}\|_F \leq \sqrt{\frac{\mu r\sigma }{n}},  \|{\mR}\|_{2,\infty} \leq  \sqrt{\frac{\mu r\sigma}{n} }\bigg\},
\end{align*}
\end{small}where $\mL_j= \mL(j s : (j+1) s - 1, :)$ is the $j$-th block of $\mL$, $\|{\mR}\|_{2,\infty}$ is the largest $\ell_2$-norm of its rows, and   $\mu,\sigma,r,s$ are parameters associated with the problem.  
\subsection{Motivation and Contributions}
\begin{figure}[!t]
		\centering		\includegraphics[width=0.8\linewidth]{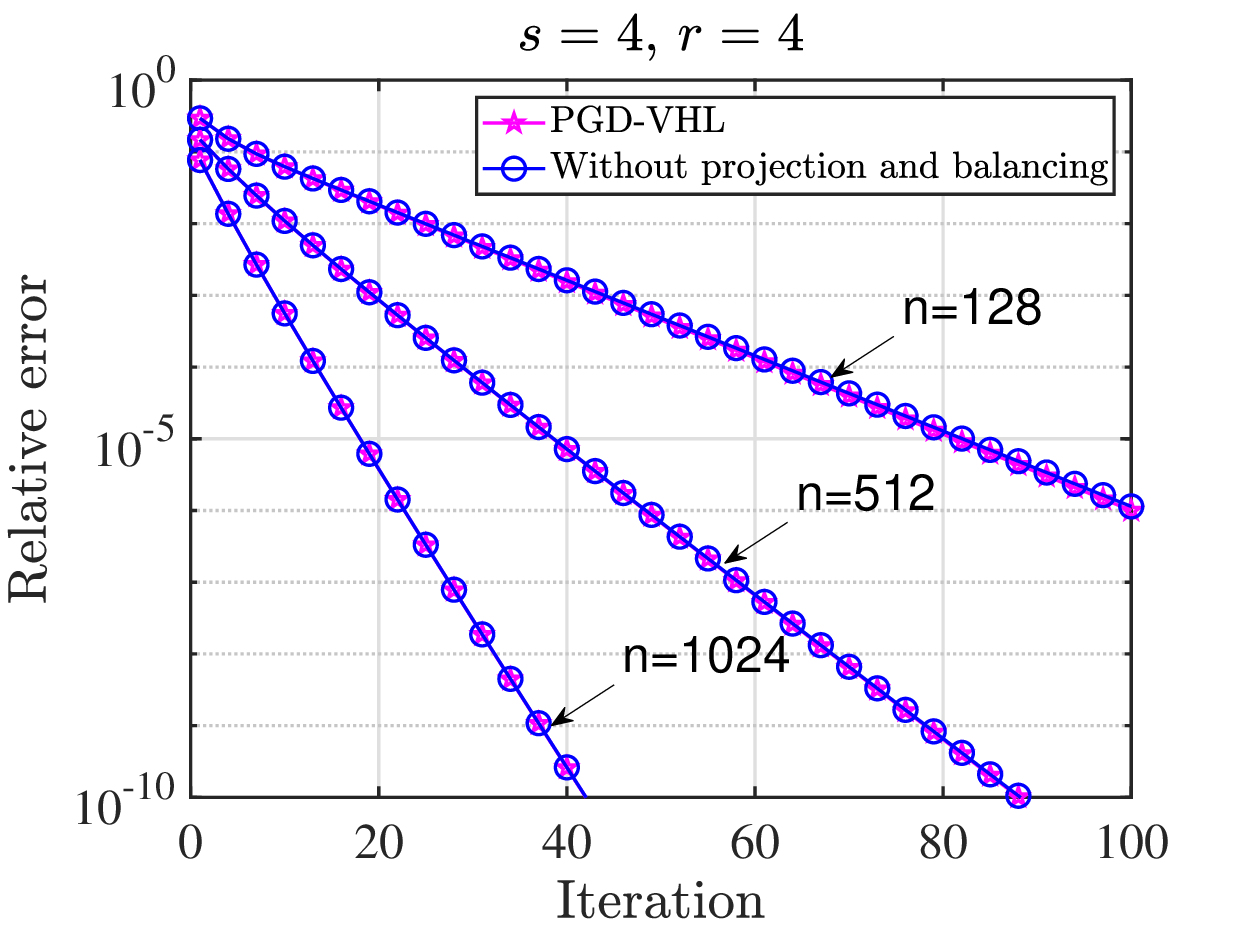} 
  		\caption{The convergence performance of PGD-VHL and vanilla gradient descent without projection and balancing, where $n$ is the length of the signal, $r$ is the number of point sources, and $s$ is the dimension of subspace the PSFs live in.}
\label{fig:intro_comp}
\end{figure}
 \begin{table*}[h!]	
		\begin{center}
			
			\caption{
   The gradient methods based on low-rank matrix factorization towards blind super-resolution.} 
			\begin{tabular}{c|c|c|c|c|c}
   \hline
				\textbf{Algorithms} & Sample complexity & Iteration complexity  & Step size & Balancing & Projection \\
	\hline			
				\hline
				PGD-VHL \cite{Mao2022}   &$O(s^2r^2\kappa^2\log^2(sn))$  &$O(s^2r^2\kappa^2\log(1/\varepsilon))$ &$O(\frac{\sigma_r(\mM_\star)}{s^2r^2\sigma_1^2(\mM_\star)})$ & $\checked$ & $\checked$ 
    \\
				VGD-VHL (Ours)  &$O(s^2r^2\kappa^3\log^2(sn))$ &$O(\kappa\log(1/\varepsilon))$ &$O(1/\sigma_1(\mM_\star))$ &$\bm{\times}$ &$\bm{\times}$ \\
				ScalGD-VHL (Ours)  &$O(s^2r^2\kappa^2\log^2(sn))$  &$O(\log(1/\varepsilon))$  &$O(1)$ &$\bm{\times}$&$\bm{\times}$ \\
    \hline
			\end{tabular}
			\label{tab:comp_algs}
		\end{center}	
	\end{table*}
However, the inclusion of two types of regularization, namely projection and the balancing penalty, proves to be unnecessary. Upon their removal, we compare the performance of this regularization-free approach with its regularized counterpart, PGD-VHL, as illustrated in Figure \ref{fig:intro_comp}. Remarkably, we find that the removal of regularization does not negatively impact performance as the two sequences converge almost identically. Moreover, removing regularization eliminates the need for tuning regularization parameters and reduces additional computation costs. These insights inspire us to consider a {new} and simpler unconstrained optimization problem without previous types of regularization:
\begin{align*}
\min_{\mF}~ f(\mF). \numberthis\label{obj function_regfree}
\end{align*} 
{Furthermore, given the burgeoning interest and successful applications of regularization-free methods across various domains \cite{Ma2021, Ma2019, Soltanolkotabi2023, Qiao2017, wu2024}, establishing a theoretical guarantee for regularization-free methods in blind supper-resolution can broaden the related applications in signal processing tasks.}
 
 In this paper, we propose two provable algorithms to solve the unconstrained optimization problem \eqref{obj function_regfree} towards blind super-resolution. We first propose a vanilla gradient descent via VHL (VGD-VHL).  
Next, we propose a scaled gradient descent variant named ScalGD-VHL to overcome the slow convergence caused by relatively small amplitudes of certain point sources. 

Our main contributions are {listed as follows:}
\begin{itemize}
    \item[1)] Two simpler but more 
efficient gradient methods based on low-rank matrix factorization and vectorized Hankel lift---VGD-VHL and ScalGD-VHL are proposed for blind super-resolution of point sources. 
In numerical results, it is shown that our algorithms converge faster and enjoy lower computational costs, with recovery performance comparable to that of prior arts. {In particular, ScalGD-VHL has the fastest convergence performance, with the lowest iteration complexity.}  

\item[2)] We establish {sharp} theoretical guarantees for our simpler methods, through {a novel less incoherence-demanding analysis} and {a convergence mechanism based on new tools}. By employing our novel and sharp analysis, we overcome the {challenges brought by the sub-optimal analysis in previous works \cite{Ma2019,Mao2022,Tong2021}.} As shown in Table \ref{tab:comp_algs}, our methods exhibit larger step size choices, lower iteration complexity, and comparable sample complexity compared to PGD-VHL \cite{Mao2022}. 

\item[3)] {Our novel less incoherence-demanding analysis 
provides sharper perturbation bounds than PGD-VHL, without the need for projection regularization. Also, this less incoherence-demanding analysis reveals deeper insights into blind super-resolution than PGD-VHL:

\emph{
    Blind super-resolution based on low-rank factorization achieves linear convergence without relying on the incoherence of its factors. 
    This problem is {more incoherence-demanding} than matrix sensing \cite{Tu2016,Ma2021}, but {less incoherence-demanding} than phase retrieval, blind deconvolution, and matrix completion \cite{Ma2019,Li2019,Tong2021,Candes2015}.} 
    
}
\end{itemize}

    





\subsection{Related work}
Recent super-resolution approaches focused on optimization-based methods, such as total-variation (TV) norm minimization \cite{Candes2013a}, atomic norm minimization (ANM) \cite{Tang2013,Bhaskar2013}, and nuclear norm minimization via enhanced Hankel matrix \cite{Chen2014}. The {blind super-resolution} problem reduces to super-resolution when the knowledge of PSFs is available. 

 Assuming the unknown PSFs live in a known low-dimensional subspace, the blind super-resolution problem can also be solved by similar optimization-based methods in super-resolution. 
 \cite{Chi2016, Yang2016,Li2020,Suliman2022} used the atomic minimization method (ANM) to solve blind super-resolution, exploiting the inherent low-dimensionality of point source signals. Inspired by the low-rank Hankel approach towards spectrally sparse signal recovery \cite{Chen2014}, a vectorized Hankel lift framework (VHL) was proposed in \cite{Chen2022}, which formulated blind super-resolution into a nuclear norm minimization problem. 
  However, the previous convex methods are computationally challenging for large-scale problems, and some fast nonconvex methods towards blind super-resolution were proposed. A fast ADMM method was developed in \cite{Ran2021} based on the atomic norm and the SDP characterization. Built on the VHL framework, \cite{Mao2022} developed a provable projected gradient descent method (PGD-VHL) based on low-rank factorization with the sample complexity as $O(s^2r^2\log^2(sn))$. Besides,  \cite{Zhu2023} proposed a fast iterative hard thresholding algorithm (FIHT-VHL) with the same sample complexity.   

In this paper, we focus on the fast nonconvex gradient methods based on low-rank factorization, which have received great interest in recent years \cite{Chi2019review}.  Numerous nonconvex gradient methods incorporate regularizations such as projection onto a constraint set and balancing terms to promote desired properties \cite{Tu2016,Zheng2016,Ge2017,Cai2018,Mao2022}. {Nonetheless, recent work has shown that these regularizations are not essential \cite{Ma2019,Chenj2020, Ma2021,Soltanolkotabi2023,Tong2021}. In what follows, we will introduce these regularization-free methods and compare with our work.}

 {\textbf{Balancing-free methods.}
 Balancing-free methods have been studied in matrix sensing \cite{Ma2021,Soltanolkotabi2023,Tong2021}, where the implicit balancing mechanism was explained. In \cite{Soltanolkotabi2023}, balancing-free gradient descent from small random initialization was studied. However, their analysis heavily depends on the Restricted Isometry Property (RIP), rendering it inapplicable to our problem that lacks this property. The balancing-free gradient methods employing spectral initialization were investigated in \cite{Ma2021,Tong2021}  and their distance metrics can be applied to our problem. Nevertheless, their analysis, in conjunction with the results of blind super-resolution in PGD-VHL  \cite{Mao2022}, leads to higher sample complexity $O(s^5r^6)$, smaller step size 
 $O(\frac{\sigma_r(\mM_\star)}{s^2r^2\sigma_1^2(\mM_\star)})$ for VGD-VHL, and slower convergence rate dependent on $\kappa$ for ScalGD-VHL. 
 In contrast, our novel analysis establishes sharper theoretical results for the proposed balancing-free methods. 
 
 }

 {\textbf{Projection-free methods.} Projection-free methods have appeared in various statistical estimation problems, 
  such as matrix sensing \cite{Tu2016,Ma2021}, phase retrieval, blind deconvolution \cite{Ma2019,Li2019}, and matrix completion \cite{Ma2019,Chenj2020}. 
  In phase retrieval, blind deconvolution \cite{Ma2019,Chen2019aGlobalpr,Li2019}, and matrix completion \cite{Ma2019,Chenj2020}, projection-free methods harness implicit regularization mechanisms to ensure the incoherence property, thereby eliminating the need for explicit projection steps. 
 In contrast, our methods eliminate the need for projection, because the blind super-resolution problem does not need the incoherence of factors to exhibit linear convergence. While gradient methods for matrix sensing in \cite{Tu2016,Ma2021} are projection-free without requirement on incoherence properties, our methods are distinct in their requirement for the incoherence property of the ground truth.  
 } 

When the target matrix is ill-conditioned, 
gradient methods based on low-rank factorization converge slowly. To overcome such issues, scaled gradient descent methods (ScaledGD) \cite{Tong2021,Tong2021a} were proposed to accelerate convergence. Similarly, the matrix constructed via VHL might be ill-conditioned in blind super-resolution, particularly when amplitudes of certain point sources are relatively small. Inspired by ScaledGD, we develop a scaled gradient descent via VHL named ScalGD-VHL to accelerate estimation. 
{However, the convergence analysis from \cite{Tong2021} cannot be generalized to blind super-resolution problem. Unlike matrix sensing in \cite{Tong2021}, blind super-resolution does not satisfy RIP. Meanwhile,
the key hammer Lemma~36, which is pivotal for the matrix completion problems in \cite{Tong2021}, does not hold for blind super-resolution. In addition, our problem does not require a projection step for the incoherence property in \cite{Tong2021}. In this work, we establish the sharp theoretical guarantee for ScalGD-VHL via a less incoherence-demanding analysis tailored for it and a convergence mechanism based on new tools.} 

\noindent \textbf{Notations}. 
We denote vectors with bold lowercase letters, matrices with bold uppercase letters, and operators with calligraphic letters. For matrix $\vZ$, we use $\vZ^T$, $\vZ^H$, $\bar{\vZ}$, $\norm{\vZ}$, and $\norm{\vZ}_F$ to denote its transpose, conjugate transpose, complex conjugate, spectral norm, and Frobenius norm, respectively. 
Define the inner product of two matrices $\vZ_1$ and $\vZ_2$ as $\la\vZ_1,\vZ_2\ra=\mathrm{trace}(\vZ_1^H\vZ_2)$. 
We denote $\otimes$ and $\odot$ as the Kronecker product and Hadamard product, respectively. 
The identity operator is denoted as $\I$. The adjoint of the operator $\A$ is denoted as $\A^*$. $\Real(\cdot)$ denotes the real part of a complex number. Let $[n]$ denote the set $\{0,1,...,n-1\}$. Denote $w_i$ as the {cardinality} of the set  $\mathcal{W}_i=
 \{(j,k) | j+k=i, 0\leq j \leq n_1-1, 0\leq k\leq n_2-1 \} $ for $i\in[n]$. We further define an operator $\D:\C^{s\times n}\rightarrow\C^{s\times n}$ as $\D(\mZ)=\begin{bmatrix}
 		\sqrt{w_0}\vz_0 & \cdots & \sqrt{w_{n-1}} \vz_{n-1}
 	\end{bmatrix}$ where $\vz_i$ denotes the $i$-th column of $\mZ$.  
  
\section{Problem formulation and Algorithms} \label{sec:pb_formul}
 \subsection{Problem formulation}
 Let $x(t)$ be a point source signal, which is a weighted sum of $r$ spikes from different locations
 \begin{align}
 	\label{eq: model form}
 	x(t) = \sum_{k=1}^r d_k\delta(t-\tau_k),
 \end{align}
 where $\delta(t-\tau_k)$ is the $k$-th spike at the continuous location of $\tau_k$, and $d_k\in\C$ is the amplitude of it. Denote the unknown point spread functions (PSFs) as $\{\bar{g}_k(t)\}_{k=1}^r$ which depends on the locations of the spikes, and we obtain the convolution between $x(t)$ and the PSFs as
 \begin{align}
 	\label{eq: sample in time domain}
 	\bar{y}(t) = \sum_{k=1}^r d_k\delta(t-\tau_k) * \bar{g}_k(t) = \sum_{k=1}^r d_k\cdot \bar{g}_k(t-\tau_k).
 \end{align} 
 {After taking the Fourier transform of \eqref{eq: sample in time domain} and sampling, we obtain the  measurements as:}
 \begin{align}
 	\label{observation model}
 	y_j = \sum_{k=1}^r d_k e^{-2\pi \imath \tau_k\cdot j} \vg_k(j),\quad j=0,\cdots n-1.
 \end{align}
Our task is to estimate the unknown PSFs $ \{\vg_k\}_{k=1}^r$, the amplitudes and locations $\{d_k, \tau_k\}_{k=1}^r$ simultaneously from \eqref{observation model}. 
 However, the previous 
 problem is ill-posed, as the number of unknowns in \eqref{observation model} is larger than the equation number $n$. As pointed out in \cite{Ahmed2014, Chi2016, Yang2016,Chen2022,Suliman2022}, this ill-posed issue is avoided by assuming that the PSFs $\{\vg_k\}_{k=1}^r$ lie in a known low-dimensional subspace, given as
 \begin{align}
 	\label{eq: low-dim of g}
 	\vg_k = \mB\vh_k,
 \end{align}
 where {$\mB\in \C^{n\times s}$ is a known low-rank matrix with $s\ll n$, whose columns represent the low-dimensional subspace to generate the PSFs $\{\vg_k\}_{k=1}^r$ with unknown coefficient vector $\vh_k\in\C^{s}$}. Let $\va_{\tau_k} = \lsb 1,e^{-2\pi \imath \tau_k},\cdots,e^{-2\pi \imath \tau_k\cdot(n-1)}\rsb^T$ and we define $\mX_\star = \sum_{k=1}^r d_k\vh_k\va_{\tau_k}^T$, then we reformulate \eqref{observation model} as 
 \begin{align}
 	\label{eq: lifting}
 	y_j = \la\vb_j\ve_j^T, \mX_\star \ra,\quad j=0,\cdots,n-1,
 \end{align}
 by combining the subspace assumption \eqref{eq: low-dim of g} and the lift trick \cite{Ahmed2014}, where
 $\vb_j\in\C^{s}$ is the $j$-th column vector of $\mB^T$, $\ve_j$ is the $j$-th standard basis of $\R^n$. Furthermore, we rewrite  \eqref{eq: lifting} into a compact form
 \begin{align}
 	\label{eq: compact form}
 	\vy = \A(\mX_\star),
 \end{align}
 where $\A:\C^{s\times n}\rightarrow \C^n$ is the linear operator and $\A^{*}(\vy) = \sum_{j=0}^{n-1}y_j\vb_j\ve_j^T$ is the corresponding adjoint linear operation. Let $\mD = \diag\left(\sqrt{w_0} ,\cdots , \sqrt{w_{n-1}} \right)$, {and one can derive that $\mD\A(\mX_\star) = \A\D(\mX_\star)$}. 
 Then we rewrite \eqref{eq: compact form} as 
 \begin{align}
 	\label{weighted measurements}
 	\mD\vy = \A\D(\mX_\star).
 \end{align}	
 After the matrix $\mX_\star$ is recovered, we apply spatial smoothing MUSIC \cite{Chen2022,Evans1982} 
 to estimate locations $\{\tau_k\}_{k=1}^r$, and solve a least square system \cite{Yang2016} to obtain amplitudes and coefficients $\{d_k,\vh_k\}_{k=1}^r$.  
 Denote $\H$ as the vectorized Hankel lift operator, 
 mapping a matrix $\mX\in\C^{s\times n}$ into an $sn_1\times n_2$ matrix, 

 \vspace{-0.3cm}
 {\begin{small} 
 \begin{align}
 	\label{vhl}
 	\H(\mX) = \begin{bmatrix}
 		\vx_0 & \vx_1 &\cdots & \vx_{n_2-1}\\
 		\vx_1 & \vx_2 &\cdots & \vx_{n_2}\\
 		\vdots& \vdots &\ddots &\vdots\\
 		\vx_{n_1-1}&\vx_{n_1}&\cdots &\vx_{n-1}\\
 	\end{bmatrix}\in\C^{sn_1\times n_2},
 \end{align}
\end{small}}

\noindent where $\vx_i\in\C^s$ is the $i$-th column of $\mX$ and $n_1+n_2 = n+1$. Let $\mM_\star=\H(\mX_\star)$ and it is a rank-$r$ matrix when $r \ll \min(sn_1, n_2)$ from \cite{Chen2022}. Then a rank constraint weighted least square problem is constructed in \cite{Mao2022} to estimate $\mX_\star$:
 \begin{align}
 	\label{pb:original obj}
 	\min_{\mX}~\frac{1}{2}\|{\mD\vy - \mathcal{A}\mathcal{D}(\mX)}\|^2 ~\text{ s.t. } \rank(\mathcal{H}(\mX)) = r.
 \end{align}
{Denote} $\G=\H\D^{-1}$, {and set} $\mM = \H(\mX) = \G\D(\mX)$. {Then one can check that $(\mathcal{I} - \G\G^*)(\mM) = \bzero$.} 
 Besides, we parameterize $\mM = \mL\mR^H$ to eliminate the rank constraint based on  Burer--Monteiro factorization \cite{Burer2003}.  Making the substitution that $\vy\leftarrow \mD\vy$  in \eqref{pb:original obj}, we reformulate \eqref{pb:original obj} into the following vectorized Hankel constraint optimization problem: 

 \begin{align}
 	\label{pb:original obj 2}
 	\min_{{\color{black}{\mL,\mR}}}~\frac{1}{2}\|{\vy - \A\G^*(\mL\mR^H)}\|_2^2 \text{~~ s.t. ~} (\I - \G\G^*)(\mL\mR^H) = \bzero.
 \end{align}
Finally, we apply a penalized version of \eqref{pb:original obj 2} to estimate $\mX_\star$: 
\begin{align*}
\min_{\mF} f(\mF)=\frac{1}{2} \|{\vy - \A\G^*(\mL\mR^H)}\|_2^2 {+} \frac{1}{2}\|{\left( \I - \G\G^* \right)(\mL\mR^H)}\|_F^2,\numberthis\label{pb:obj function}
\end{align*}
where $\mF=\begin{bmatrix}
 	\mL^H &\mR^H
 \end{bmatrix}^H\in\C^{(sn_1+n_2)\times r}$. {The weighting parameter in \eqref{pb:obj function} is fixed to assure $f(\mF)$ is an unbiased estimator of $f_0(\mF)$ in \eqref{eq:f0} such that $\E(f(\mF))=f_0(\mF)$.}
Unlike PGD-VHL \cite{Mao2022}, we derive a simpler optimization problem without balancing and projection. It is safe to remove them from the performance in Fig. \ref{fig:intro_comp}. In addition, such regularization is used for the ease of analysis of PGD-VHL \cite{Mao2022}, while both are unnecessary for the analysis of our algorithms. 
\subsection{Algorithms}\label{subsec:alg}
 We propose two gradient descent methods to solve the previous unconstrained problem \eqref{pb:obj function}. Both algorithms start from the spectral initialization $\T_r(\G\A^*(\vy))$ where $\T_r(\cdot)$ 
 performs top-$r$ SVD of a matrix and $\A^{*}$ is the adjoint of $\A$.  The first algorithm is vanilla gradient descent via vectorized Hankel lift, and we name it VGD-VHL, seeing Alg. \ref{alg:VGD}.
The updating rule of VGD-VHL is:
 \begin{align*}
&\mL^{k+1}=\mL^{k}-\eta\nabla_{\mL} f(\mF^{k}),\\
&\mR^{k+1}=\mR^{k}-\eta\nabla_{\mR} f(\mF^{k}).
\end{align*}
Under Wirtinger calculus, the gradient of  $f(\vF)$ is:
\begin{align*}
\nabla_{\mL} f(\mF)&{=}\big(\G\A^*(\A\G^*(\mL\mR^H)-\vy)+(\I-\G\G^*)(\mL\mR^H)\big)\mR,\\
\nabla_{\mR} f(\mF)&{=}\big(\G\A^*(\A\G^*(\mL\mR^H)-\vy)+(\I-\G\G^*)(\mL\mR^H)\big)^H\mL.
\end{align*}   
Inspired by \cite{Tong2021}, we propose a scaled gradient descent variant named ScalGD-VHL that accelerates convergence when the vectorized Hankel lifted matrix $\mM_\star=\H(\mX_\star)$ is ill-conditioned, seeing Alg. \ref{alg:ScaleGD}. ScalGD-VHL proceeds as follows:
  \begin{align*}
&\mL^{k+1}=\mL^{k}-\eta\nabla_{\mL} f(\mF^{k})\big((\mR^{k})^H\mR^{k}\big)^{-1},\\
&\mR^{k+1}=\mR^{k}-\eta\nabla_{\mR} f(\mF^{k})\big((\mL^{k})^H\mL^{k}\big)^{-1},
\end{align*}
which preconditions the gradient, enabling a better search direction and larger step size.
\begin{algorithm}[t]
\caption{{VGD-VHL:} Vanilla Gradient Descent via Vectorized Hankel Lift }
\label{alg:VGD}
\begin{algorithmic} 
\Statex \textbf{Initialization:} $\vM^{0}=\T_r \(\G\A^*(\vy)\)=\vU^0\vSS^0(\vV^0)^H$,\\
$\mL^0=\vU^0(\vSS^0)^{\frac{1}{2}}$,  $\mR^0=\vV^0(\vSS^0)^{\frac{1}{2}}$.
\For{$k=0,1,\cdots,K$}\\
\quad $\mL^{k+1}=\mL^{k}-\eta\nabla_{\mL} f(\mL^{k},\mR^{k})$\\
\quad$\mR^{k+1}=\mR^{k}-\eta\nabla_{\mR} f(\mL^{k},\mR^{k})$
\EndFor
\Statex \textbf{Output:} $(\mL^{K},\vR^K)$, $\mX^K = \D^{-1}\G^*(\mL^{K}(\mR^K)^{H})$.
\end{algorithmic}
\end{algorithm}

\begin{algorithm}[t]
\caption{{ScalGD-VHL:} Scaled Gradient Descent via Vectorized Hankel Lift}
\label{alg:ScaleGD}
\begin{algorithmic} 
\Statex \textbf{Initialization:} $\vM^{0}=\T_r \(\G\A^*(\vy)\)=\vU^0\vSS^0(\vV^0)^H$,\\
$\mL^0=\vU^0(\vSS^0)^{\frac{1}{2}}$,  $\mR^0=\vV^0(\vSS^0)^{\frac{1}{2}}$.
\For{$k=0,1,\cdots,K$}\\
\quad $\mL^{k+1}=\mL^{k}-\eta\nabla_{\mL} f(\mL^{k},\mR^{k})\big((\mR^{k})^H\mR^{k}\big)^{-1}$\\
\quad$\mR^{k+1}=\mR^{k}-\eta\nabla_{\mR} f(\mL^{k},\mR^{k})\big((\mL^{k})^H\mL^{k}\big)^{-1}$
\EndFor
\Statex \textbf{Output:} $(\mL^{K},\vR^K)$, $\mX^K = \D^{-1}\G^*(\mL^{K}(\mR^K)^{H})$.
\end{algorithmic}
\end{algorithm}

\section{Theoretical guarantees} \label{sec:theoretical-results}
{In this section, we first introduce basic preliminaries. Then we present the challenges to establishing sharp theoretical guarantees and give the theoretical results for our algorithms.} 
\subsection{Preliminaries}
We make the assumption that $\mM_\star = \H(\mX_\star)$ is $\mu_1$-incoherent: 
 \begin{assumption}
 	\label{assumption1}
 	Let $\mM_\star= \mU_\star\vSS_\star\mV_\star^H$ be the singular value decomposition of $\mM_\star$, where $\mU_\star\in\C^{sn_1\times r}, \vSS_\star\in\R^{r\times r}$ and $\mV_\star\in\C^{n_2\times r}$. Denote $\mU_\star^H = \begin{bmatrix}
 		(\mU_{\star}^0)^H &\cdots &(\mU_{\star}^{n_1-1})^H
 	\end{bmatrix}^H$, where the $j$-th block of $\mU_\star$ is $\mU_{\star}^j= \mU_\star(j s : (j+1) s - 1, :)$ for $j=0,\cdots n_1-1$. We say the matrix $\mM_\star$ is $\mu_1$-incoherent when $\mU_\star$ and $\mV_\star$ satisfy that
 	\begin{align*}
 		\max_{0\leq j\leq n_1-1} \|{\mU_\star^j}\|_F^2 \leq \frac{\mu_1 r}{n} \text{ and }\max_{0\leq k \leq n_2-1} \|{\ve_k^T\mV_\star}\|_2^2 \leq \frac{\mu_1 r}{n}
 	\end{align*}
 	for some positive constant $\mu_1 $.
 \end{assumption}
This incoherence property has also appeared in \cite{Cai2018,Candes2009,Chen2022,Chen2014}. It has been demonstrated in \cite{Chen2022} that when the minimum wrap-up distance between the locations of point sources is greater than approximately $2/n$,  the matrix $\vM_\star=\H(\mX_\star)$ is $\mu_1$-incoherent. 

 \begin{assumption} 
 	\label{assumption2}
 	The column vectors $\{\vb_i\}_{i=0}^{n-1}$ of $\mB^H$ are independently and identically generated from a distribution $F$, equipped with the following conditions:
 	\begin{align}
 		\label{eq: isotropy property}
 		\E[\vb_i\vb_i^H] &= \mI_s, \quad i=0,\cdots, n-1,\\
 		\label{eq:incoherence}
 		\max_{0\leq \ell\leq s-1}  | \vb_i[\ell]|^2 &\leq \mu_0, \quad i=0,\cdots, n-1.
 	\end{align}
 \end{assumption}
    The previous properties 
    are standard in blind super-resolution problem \cite{Chi2016,Yang2016,Chen2022,Suliman2022,Mao2022}. We also assume $\mu_0 s\geq 1$ as discussed in \cite{Chi2016,Yang2016} for ease of analysis, which can be ensured by choosing $\mu_0\geq1$.

    In Alg. \ref{alg:VGD} and Alg. \ref{alg:ScaleGD}, the error between the current estimate $\mX$ and the target data matrix $\mX_\star$ can be bounded by
    \begin{align*}
       \|\mX{-}\mX_\star\|_F^2=\|\D^{-1}\G^{*}\(\mL\mR^H{-}\mM_\star\)\|_F^2\leq\|\mL\mR^H{-}\mM_\star\|_F^2, \numberthis \label{eq:estimate_upbd}
    \end{align*}
as $\|\G^{*}\|\leq 1$, $\|\D^{-1}\|\leq 1$. 
\
Define 
    \begin{align*}
        f_0(\mF)=\frac{1}{2}\|\mL\mR^H-\mM_\star\|_F^2. \numberthis \label{eq:f0}
    \end{align*} 
Consequently, establishing the linear convergence of $f_0(\mF)$ is enough to make the estimate $\mX$ approach the target matrix $\mX_\star$. Note that we don't rely on establishing the convergence of a distance metric designed for factors matrices as in \cite{Cai2018,Mao2022}. 
    
    Besides, we rewrite that $f(\mF)=f_0(\mF)+f_1(\mF)$ and denote $f_1(\mF)$ as 
\begin{align*}
    f_1(\mF)=\frac{1}{2} \|{\vy - \A\G^*(\mL\mR^H)}\|^2- \frac{1}{2}\|\G\G^*(\mL\mR^H-\mM_\star)\|_F^2\numberthis\label{eq:f1}.
\end{align*}

To help analyze the convergence of ScalGD-VHL, we define a scaled norm inspired by \cite{Zhang2023}, and generalize it to the asymmetric matrix case. {Unlike VGD-VHL}, the convergence analysis for ScalGD-VHL relies on results under the scaled norm, not the Frobenius norm.  
\begin{definition}[{The scaled norm}]
  For $\mV\in\C^{(sn_1+n_2)\times r}$ and $\mV=\begin{bmatrix}
    \mV_L \\
    \mV_R
\end{bmatrix}$ where $\mV_L\in\C^{sn_1\times r}$ and $\mV_R\in\C^{n_2\times r}$, define
\begin{align*}
    \|\mV\|_{\mL,\mR}=\sqrt{\|\mV_L(\mR^{H}\mR)^{\frac{1}{2}}\|_F^2+\|\mV_R(\mL^{H}\mL)^{\frac{1}{2}}\|_F^2},
\end{align*}
and the corresponding dual norm:
\begin{align*}
    \|\mV\|_{\mL,\mR}^{*}=\sqrt{\|\mV_L(\mR^{H}\mR)^{-\frac{1}{2}}\|_F^2+\|\mV_R(\mL^{H}\mL)^{-\frac{1}{2}}\|_F^2}.
\end{align*}  
\end{definition}
\subsection{{Challenges} and main results}
{We first introduce the challenges to establish sharp theoretical guarantees for our methods in blind super-resolution. 

A straightforward combination of the analysis in balancing-free gradient methods \cite{Ma2021,Tong2021} and the results in PGD-VHL \cite{Mao2022} leads to worse theoretical results for VGD-VHL and ScalGD-VHL. For VGD-VHL, the previous analysis yields {higher sample complexity} as $O(s^5r^6)$ and {smaller step size as $O(\frac{\sigma_r(\mM_\star)}{s^2r^2\sigma_1^2(\mM_\star)})$}. The high sample complexity stems from its dependence on the step size, which is $O(1/\eta^2)$ according to the analysis in \cite{Ma2021}. 
Thus a small step size $\eta$ as in \cite{Mao2022} with $\eta = O(\frac{\sigma_r(\mM_\star)}{s^2r^2\sigma_1^2(\mM_\star)})$ directly leads to a high sample complexity. For ScalGD-VHL, the previous analysis yields {a slow convergence rate dependent on $\kappa$}. This is due to the interplay between the incoherence property of the factors and the preconditioner within ScalGD-VHL. This interaction dictates a step-size selection of $O(1/\kappa^2)$, yielding a slow convergence rate that correlates with $\kappa$.


We overcome the previous challenges via a convergence mechanism based on new tools and a novel less incoherence-demanding analysis.
The convergence mechanism is based on the Polyak-Łojasiewicz (PL) inequality \cite{Polyak1963,Lojasiewicz1963} and smoothness condition under the (scaled) Frobenius norm, seeing Lemma \ref{lem:pl-fro} and \ref{lem:pl-scale},  Lemma~\ref{lem:smoothness-fro} and \ref{lem:smoothness-scale} for details. This provides new perspectives for the convergence of balancing-free gradient methods in asymmetric matrix estimations. Besides, our less incoherence-demanding analysis removes the dependence on the incoherence of factors, which derives sharper perturbation bounds than \cite{Mao2022}, seeing proofs in Lemma~\ref{lem:upbd-perturbgd-fro} and \ref{lem:upbd-gd-scale} for our analysis. Last, we emphasize that the {sharp} perturbation bound $ \|\nabla f_1(\mF)\|_{\mL,\mR}^{*}\lesssim \delta\|\nabla f_0(\mF)\|_{\mL,\mR}^{*}$ under the (dual) scaled norm in Lemma~\ref{lem:upbd-gd-scale} is {firstly established} for blind super-resolution, which is new to the best of our knowledge.}
   
    Next, we provide the theoretical guarantees for VGD-VHL and ScalGD-VHL. {The proofs of the following theorems are deferred to Appendix \ref{apd:pf-VGD} and Appendix \ref{apd:pf-ScaleGD}. } 
\begin{theorem}[Exact recovery of VGD-VHL]\label{thm:VGD-VHL}
    With probability at least $1-c_1(sn)^{-c_2}$, the iterate in Alg. \ref{alg:VGD}
satisfies    
    \begin{align*}
        \|\mX^k-\mX_\star\|_F^2\leq \Big(1-\frac{\eta\sigma_r(\mM_\star)}{28}\Big)^{k}\sigma_r^2(\mM_\star)
    \end{align*}
    provided $\eta\leq\frac{1}{25\sigma_1(\mM_\star)}$ and $n\geq C_1\delta_0^{-2} \mu_0^2\mu_1s^2r^2\kappa^3\log^2(sn)$, where $\delta_0\leq\frac{2}{25}$, $C_1$ is a universal constant and $\kappa=\frac{\sigma_1(\mM_\star)}{\sigma_r(\mM_\star)}$.
\end{theorem}
\begin{theorem}[Exact recovery of ScalGD-VHL]\label{thm:ScalGD-VHL}
   With probability at least $1-c_1(sn)^{-c_2}$, the iterate in Alg. \ref{alg:ScaleGD}
satisfies    
    \begin{align*}
        \|\mX^k-\mX_\star\|_F^2\leq \Big(1-\frac{\eta}{25}\Big)^{k}\sigma_r^2(\mM_\star)
    \end{align*}
    provided $\eta\leq\frac{1}{4}$ and $n\geq C_2\delta_0^{-2} \mu_0^2\mu_1s^2r^2\kappa^2\log^2(sn)$, where $\delta_0\leq\frac{1}{10}$, $C_2$ is a universal constant.
\end{theorem}
\begin{remark}[Step-size strategy]
We prove a larger step-size strategy $\eta=O(1/\sigma_1(\mM_\star))$ for VGD-VHL and $\eta=O(1)$ for ScalGD-VHL, which is independent of $s$ and $r$. This is a sharp improvement compared to the conservative step-size strategy $O(\frac{\sigma_r(\mM_\star)}{s^2r^2\sigma_1^2(\mM_\star)})$ in PGD-VHL \cite{Mao2022}. Larger step-size choices allow for more flexibility and efficiency in the optimization process, and the theoretical support for this provides a solid foundation for the practical implementation. 
\end{remark}

\begin{remark}[Iteration complexity] \label{remark:itercom_acc}
    To {achieve} the $\varepsilon$ accuracy  $\|\mX^k-\mX_\star\|_F\leq\varepsilon\sigma_r(\mM_\star)$, the iteration complexity is $O(\kappa\log(1/\varepsilon))$ for VGD-VHL and $O(\log(1/\varepsilon))$ for ScalGD-VHL, which are lower than $O(s^2r^2\kappa^2\log(1/\varepsilon))$ in PGD-VHL. In numerical simulations \ref{Sim:sub-convergence}, VGD-VHL and PGD-VHL need almost the same iterations to {achieve} the same accuracy. This implies the theoretical iteration complexity in PGD-VHL may not be sufficiently sharp, while we provide a sharper one for VGD-VHL. ScalGD-VHL requires the fewest iterations to reach the same accuracy and its iteration complexity is independent of $\kappa$. 
\end{remark}
\begin{table*}[h!]	
{
		\begin{center}			
			\caption{
   The comparisons on the demand of incoherence for different low-rank matrix recovery problems.} 
   {
     \begin{tabular}{c|c|c}
   \hline
   \textbf{Problem} &\textbf{Incoherence}  & \textbf{Measurements}
   \\ \hline \hline			
{Matrix sensing} &{None}   &{ $y_i=\langle\mA_i,\mM_\star\rangle$, for $i\in[m]$} 
\\ 
{Matrix completion } & ground truth $\mM_\star$ and 
iterates $\{\mL^k, \mR^k\}$&{$Y_{i,j}=\langle\ve_{i}\ve_{j}^T,\mM_\star\rangle$ for $(i,j)$ in the set $\Omega$} 
   \\  
  {Phase retrieval and blind deconvolution} &  ground truth $\mM_\star$ and 
  iterates $\{\bm{h}^k,\vx^k\}$ &{$y_{i}=\langle\vb_i\va_i^H,\mM_\star\rangle
   $ for $i\in [m]$}
    \\
    {Blind super-resolution via VHL (this work)}
    &  ground truth $\mM_\star$ & {$y_{i}=\frac{1}{w_i}\langle\H(\vb_i\ve_i^T),\mM_\star\rangle$ for $i\in [n]$}
      \\ 
   \hline 

			\end{tabular}
   }
			\label{tab:comp-pbincoh}
   \end{center}
   }
\end{table*}
 
{\begin{remark}[Initialization]
   The balanced initialization in Alg. \ref{alg:VGD} and Lemma~\ref{lem:spec-init} is a necessary condition for VGD-VHL to ensure fast convergence, as indicated by our analysis in Appendix \ref{apd:pf-VGD} and preliminary simulations. On the other hand, ScalGD-VHL demonstrates resilience to unbalanced initialization, thereby ensuring a fast convergence rate.  This stems from the fact that the preconditioner within ScalGD-VHL implicitly maps unbalanced factors to their balanced counterparts, which is its implicit balancing mechanism as indicated in our analysis and \cite{Tong2021}.
\end{remark}
}

{
\section{Discussion}
}
{
First, this section is dedicated to a detailed discussion on how the incoherence demands of the blind super-resolution problem compare with those of other low-rank matrix recovery problems \cite{davenport2016overview}, based on the low-rank factorization framework. Next, we investigate whether our analysis, which requires less incoherence, can be applied to other problems.

\vspace{5pt}
\subsection{Comparisons on the demand for incoherence }

We compare the demand of incoherence and the measurements for our problem with other low-rank matrix recovery problems via low-rank factorization, such as matrix sensing \cite{Tu2016,Ma2021}, matrix completion\cite{Tong2021,Ma2019,Sun2016}, phase retrieval \cite{Chen2019aGlobalpr,Candes2015}, and blind deconvolution  \cite{Li2019, Ma2019} in Table \ref{tab:comp-pbincoh}.

 Our problem is more incoherence-demanding than matrix sensing \cite{Tu2016,Ma2021} as it does not satisfy RIP and requires the incoherence property of the ground truth $\mM_\star$. For phase retrieval \cite{Chen2019aGlobalpr,Candes2015} and blind deconvolution  \cite{Ma2019,Li2019}, they need the incoherence to guarantee small $\|\mB\vh^k\|_{\infty}$ throughout the iteration, which results from the bilinear measurements  $y_{i}=\langle\vb_i\va_i^H,\mM_\star\rangle$ for $i\in[m]$. 
Unlike them, our problem avoids the requirement for incoherent iterates. The reason is that the measurements of our problem after Hankel lifting are $y_{i}=\frac{1}{w_i}\langle\H(\vb_i\ve_i^T),\mM_\star\rangle$ for $i\in[n]$, which are not bilinear and sample off-diagonally across two dimensions of the matrix. Matrix completion \cite{Tong2021,Ma2019,Sun2016} requires the incoherence of factors to ensure small $\|\mL^k\|_{2,\infty}$ and $\|\mR^k\|_{2,\infty}$ throughout the iteration due to the element-wise measurements $Y_{i,j}=\langle\ve_{i}\ve_{j}^T,\mM_\star\rangle$ for $(i,j)$ in the set $\Omega$. In contrast, our problem does not require this type of incoherence for iterates, as the measurements in our case are less structured and exhibit more randomness compared to the element-wise sampling in matrix completion.
 
 

\subsection{Extension to other problems}
Our less incoherence-demanding analysis can be extended to a new type of low-rank matrix sensing that doesn't satisfy RIP. Low-rank matrix sensing problem is to estimate the  rank-$r$ matrix $\mM_\star\in \C^{n_1\times n_2}$ from
\begin{align*}
    y_i=\langle\mA_i,\mM_\star\rangle,~~i=0,\cdots,m-1,
\end{align*}
where $\mA_i\in\C^{n_1\times n_2}$ is the $i$-th sensing matrix. Rewrite the previous measurements as $\vy=\tilde{\A}(\mM_\star)$ where $\tilde{\A}(\mM)=\{\langle\mA_i,\mM\rangle\}_{i\in[m]}$. Based on the low-rank factorization, one can formulate the following optimization problem: 
\begin{align*}
       \min_{\mL\in\C^{n_1\times r},\mR\in\C^{n_2\times r}}~~\frac{1}{2}\|\vy-\tilde{\A}(\mL\mR^H)\|_F^2, 
   \end{align*}
  and apply the gradient descent method to solve the problem. 
  When $\tilde{\A}$ satisfies the following conditions \eqref{RIPless}, it is straightforward to provide the linear convergence guarantees of the gradient method by extending our analysis: First, we can establish a similar sharp perturbation bound as in Lemma~\ref{lem:upbd-perturbgd-fro} equipped with conditions \eqref{RIPless}, and then following the routes in Appendix \ref{apd:pf-VGD}, the linear convergence result can be derived. We omit the proof for brevity.

\begin{theorem}  
Provided the operator $\tilde{\A}$ satisfies
\begin{align*}
    \|\tilde{\A}\|\leq t,~~\|\P_T(\tilde{\A}^{*}\tilde{\A}-\I)\P_T\|\leq \varepsilon,\numberthis\label{RIPless}
\end{align*}
where $T$ denotes the tangent space of $\mM_\star$. For vanilla gradient descent as in Alg. \ref{alg:VGD}, if the balanced initialization satisfies $\|\mL^0(\mR^0)^H-\mM_\star\|_F\leq \varepsilon\sigma_r(\mM_\star)/t$, then the iterates exhibit linear convergence as  
\begin{align*}
    \|\mL^k(\mR^k)^H-\mM_\star\|_F^2\leq \(1-c/\kappa\)^k\sigma_r^2(\mM_\star),
\end{align*}
where $c$ and $\varepsilon$ are constants. 
\end{theorem}

To the best of our knowledge, this type of matrix sensing problem {without RIP} is {novel} to the current literature. 
In particular, blind super-resolution in our work can be seen as {a special case} of this problem. One can obtain this directly via redefining $\tilde{\A}=\A\G^{*}$ and checking the conditions in \eqref{RIPless}. In addition, some variants of blind super-resolution such as simultaneously blind demixing and super-resolution of point sources in \cite{Wang2024a} also satisfy the conditions in \eqref{RIPless} and belong to this type of problem. 
It is interesting to study whether other practical problems beyond blind super-resolution can be seen as a special case of this problem. Meanwhile, it is also interesting to study what statistical properties that $\tilde{\A}$ should satisfy to obtain the conditions in \eqref{RIPless}.
}

\section{Numerical Simulations} \label{sec:numerical}
\begin{figure*}[!t]
		\centering
  \includegraphics[width=0.20\linewidth]{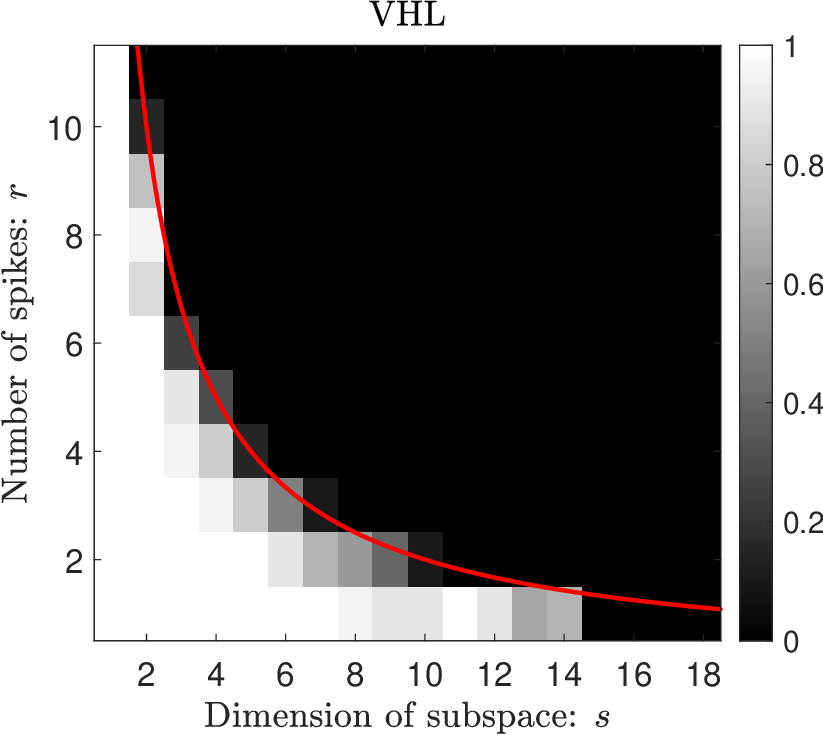}  \hfill
		\includegraphics[width=0.20\linewidth]{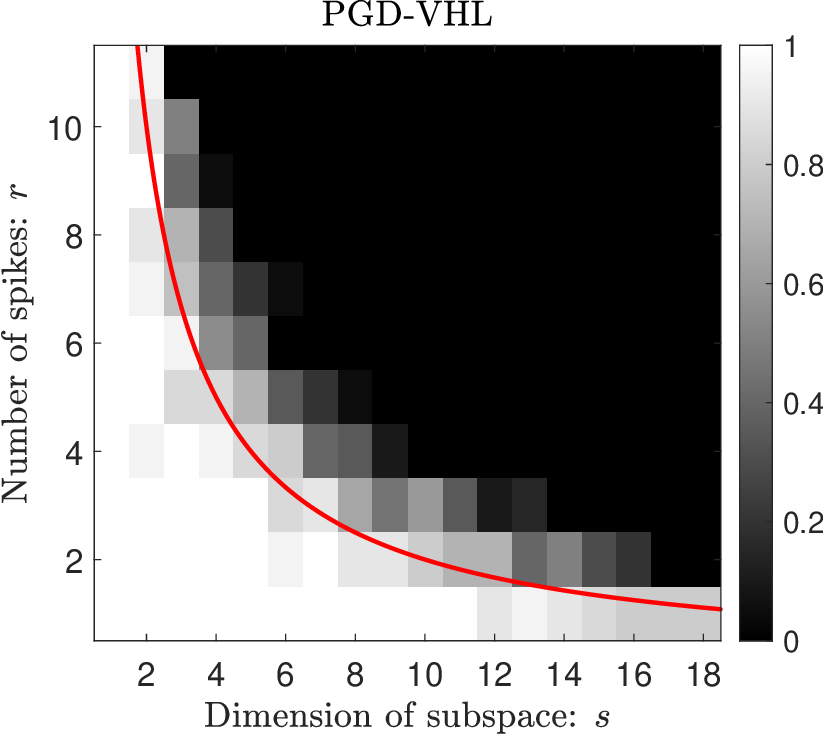}  \hfill
  \includegraphics[width=0.20\linewidth]{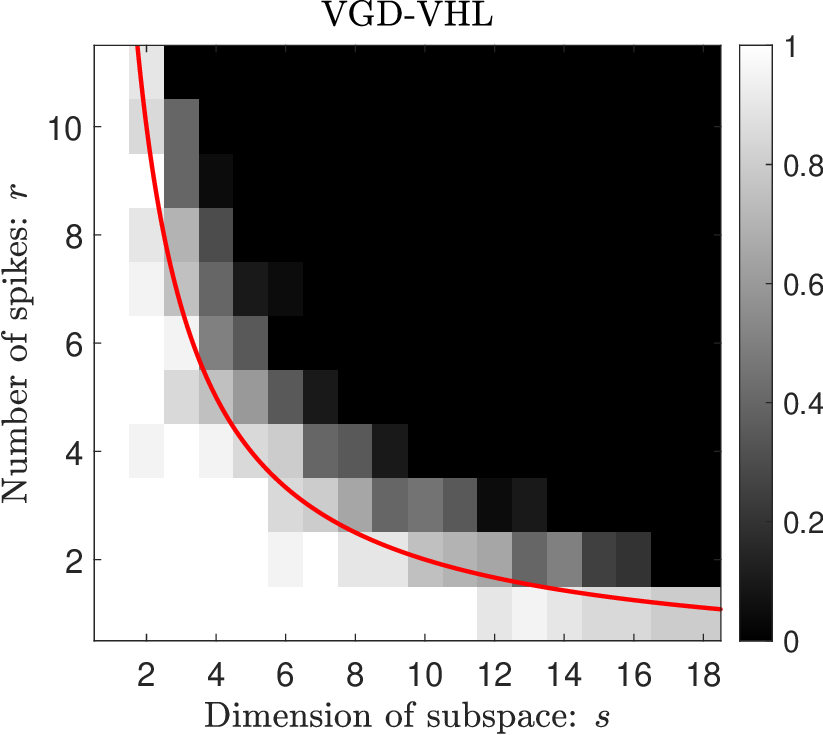}  \hfill
  \includegraphics[width=0.20\linewidth]{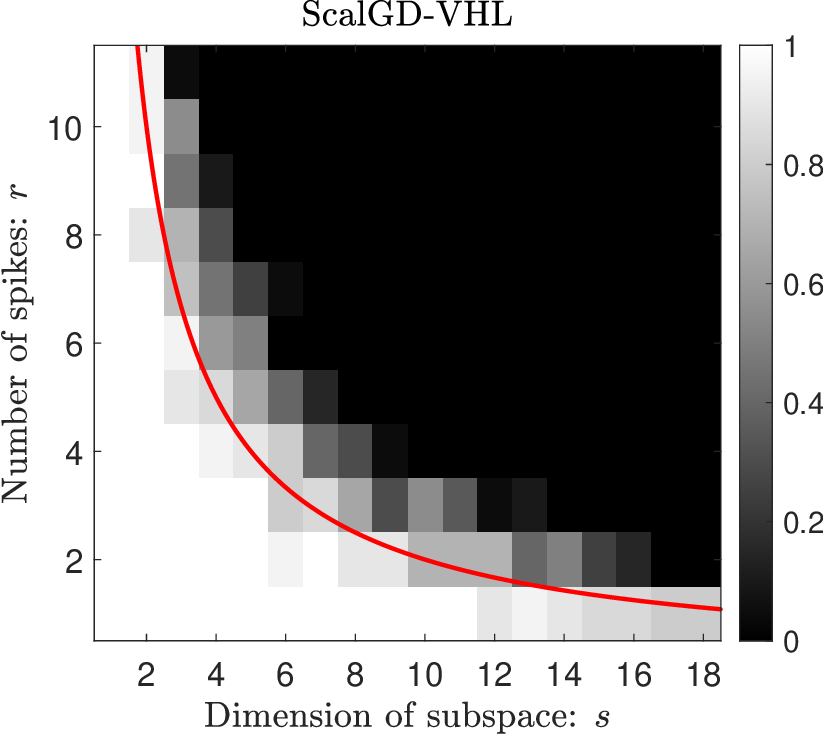} \hfill
		\caption{The comparisons of VHL, PGD-VHL, VGD-VHL, and ScalGD-VHL in terms of phase transitions. The red curve plots $sr=20$.}
		\label{fig:phasetrans_withoutsep}
  \end{figure*}
  
Extensive numerical simulations are carried out to demonstrate the performance of our algorithms \footnote{Our code is available at \url{https://github.com/Jinshengg/SimplerGDs-VHL}.}. We first test the phase transition performance of VGD-VHL and ScalGD-VHL in comparison to prior arts.
Then the convergence performance of our algorithms 
versus iterations is presented in \ref{Sim:sub-convergence}. 
Besides, we provide runtime comparisons 
 for our algorithms and PGD-VHL under different 
 settings in \ref{subsec:sim-runtime}. 
 Lastly, we demonstrate the performance of our algorithms towards a real-world case, the joint delay-Doppler estimation problem in \ref{subsec:real2D}. {All simulations are run in MATLAB R2019b on a 64-bit Windows machine with 
 Intel CPU i9-10850K at 3.60 GHz and 16GB RAM.} 
\subsection{Phase transitions} \label{subsec:phasetran}
In this subsection, we present the phase transition performance of convex method VHL \cite{Chen2022}, PGD-VHL \cite{Mao2022}, VGD-VHL, and ScalGD-VHL. 
The target data matrix 
is generated by $\mX_\star = \sum_{k=1}^r d_k\vh_k\va_{\tau_k}^T$, where 
 $\{\vh_k\}_{k=1}^r$ are generated from standard Gaussian distribution with normalization, $\{d_k\}_{k=1}^r$ are chosen as $d_k=(1+10^{c_k})e^{-i\phi_k}$ with $c_k$ uniformly taken from $[0,1]$ and $\phi_k$ uniformly taken from $[0,2\pi)$, and the locations of point sources $\{\tau_k\}_{k=1}^r$ are randomly chosen from $[0,1)$. The columns of the low-dimensional subspace $\mB$ are chosen randomly from the DFT matrix. 
 
The step-size is set as $\eta=0.4/\|\mM^0\|$ for PGD-VHL and VGD-VHL, and $\eta=0.4$ for ScalGD-VHL. 
The termination condition for three 
gradient methods is $\|\mX^{k+1}-\mX^k\|_F/\|\mX^k\|_F\leq10^{-7}$ or the maximum number of iterations is reached. The convex method VHL is implemented using CVX. We set $n=64$ as fixed, with varying $s$ and $r$ for phase transition testing. 
 We run 30 random trials for each setting and record the success rate. A test is declared to be successful if $\|\mX^k-\mX_\star\|/\|\mX_\star\|_F\leq10^{-3}$. In Fig. \ref{fig:phasetrans_withoutsep}, we observe that all nonconvex methods outperform the convex method VHL \cite{Chen2022}. In addition, our simpler gradient methods, VGD-VHL and ScalGD-VHL achieve recovery performance comparable to that of PGD-VHL, while eliminating two types of regularization: projection and balancing. 

\subsection{Convergence performance}\label{Sim:sub-convergence}
\begin{figure}[!t]
		\centering
	  \includegraphics[width=0.70\linewidth]{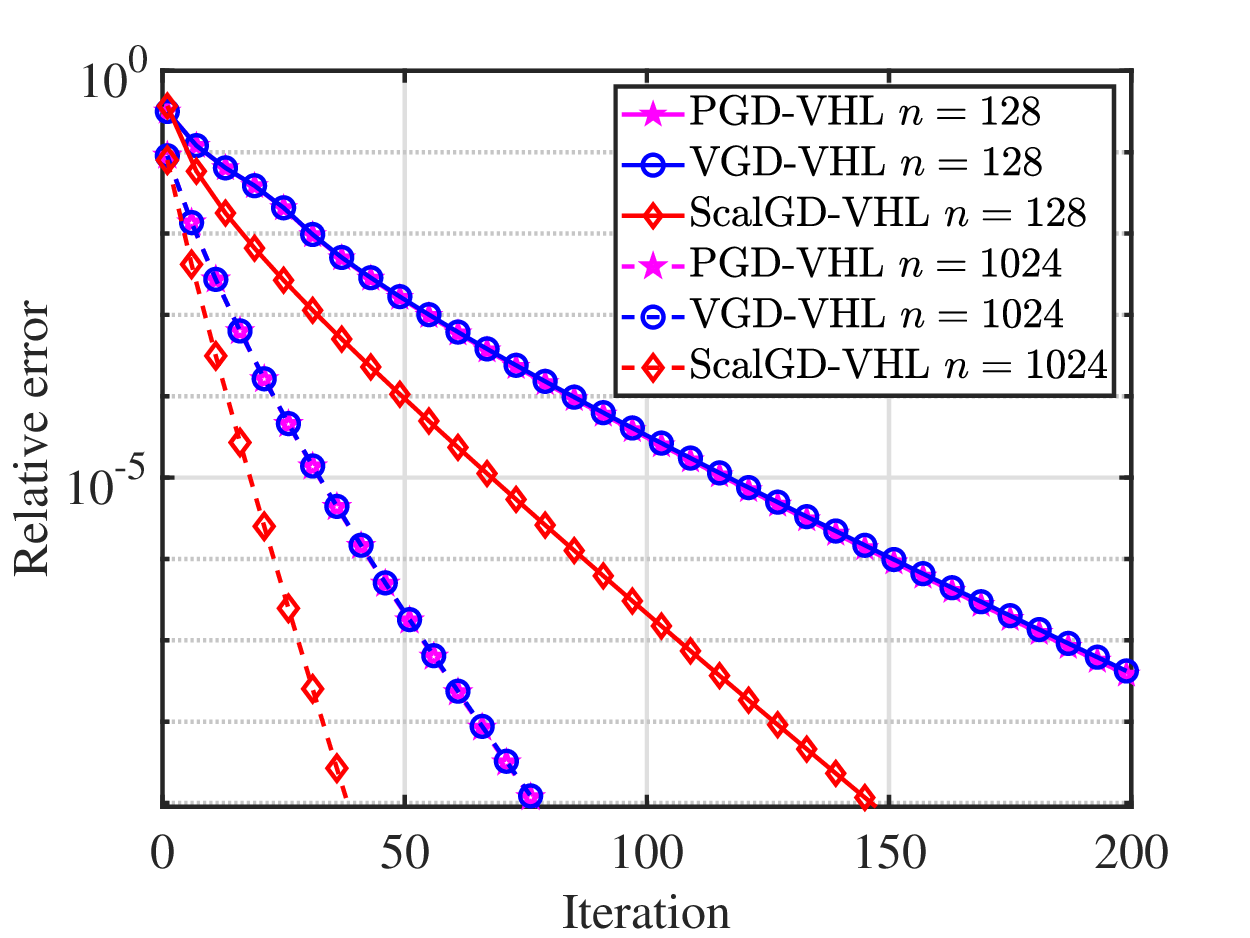}	
  \caption{The relative errors of PGD-VHL, VGD-VHL, and ScalGD-VHL versus the iteration number.}
		\label{fig:iter_err}
\end{figure}

The convergence performance of PGD-VHL, VGD-VHL, and ScalGD-VHL versus iterations is presented in this subsection. 
The relative error in our experiments refers to $\|\mX^k-\mX_\star\|_F/\|\mX_\star\|_F$. 

In the first experiment, we test the convergence performance of the algorithms for the length of the signal $n=128$ and $n=1024$ respectively, with fixed $s=r=4$. 
We run 20 random trials for each case and record the average relative error versus iterations. 
Other arguments are the same as in \ref{subsec:phasetran}. From Fig. \ref{fig:iter_err}, it can be observed that ScalGD-VHL exhibits the fastest convergence, requiring fewer iterations than VGD-VHL and PGD-VHL to attach the same accuracy.  
 Although both PGD-VHL and VGD-VHL exhibit similar convergence performance in terms of iterations, the theoretical iteration complexity of VGD-VHL is lower than that of PGD-VHL. This suggests that the theoretical guarantees provided by PGD-VHL \cite{Mao2022} may not be as sharp as the ones we provide for our algorithms. 

In the second experiment, we investigate the convergence performance of VGD-VHL and ScalGD-VHL when the vectorized  Hankel matrix $\mM_\star=\H(\mX_\star)$ is ill-conditioned. The target data matrix 
is generated by $\mX_\star = \sum_{k=1}^r d_k\vh_k\va_{\tau_k}^T$ where the locations of point sources $\{\tau_k\}_{k=1}^r$ are randomly sampled from $\{1/n,2/n,...,1\}$ and the amplitudes $\{d_k\}_{k=1}^r$ are set in three cases: all ones, linearly distributed from $[1/5,1]$ and from $[1/15,1]$. 
Consequently, the lifted matrix $\mM_\star$ has the condition number as $\kappa=1,5,15$ respectively, which can be verified from the Vandermonde decomposition of $\mM_\star=\H(\mX_\star)$ as shown in \cite{Chen2022}. We run VGD-VHL and ScalGD-VHL 20 times for each case and record the average relative error. In Fig. \ref{fig:kappa_convergence}, we observe that the convergence performance of ScalGD-VHL is independent of the condition number $\kappa$, while VGD-VHL converges slower as $\kappa$ increases. 
\begin{figure}[!t]
		\centering
	  \includegraphics[width=0.7\linewidth]{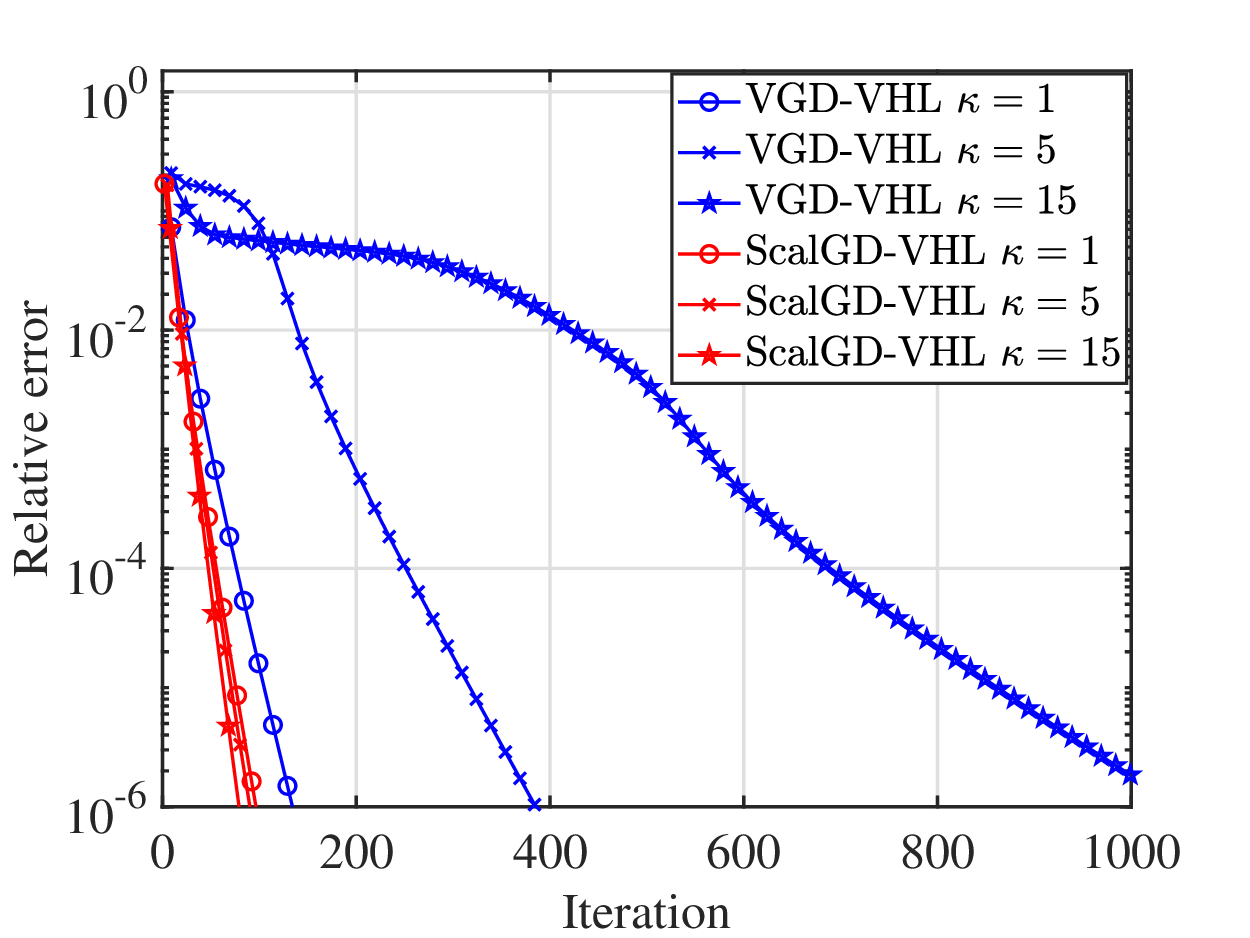}	
  \caption{The comparisons of VGD-VHL and ScalGD-VHL under different condition numbers of $\vM_\star$ as  $\kappa=1, 5, 15$.} 

  \label{fig:kappa_convergence}
\end{figure}

\begin{figure}
    \centering
    \includegraphics[width=0.67\linewidth]{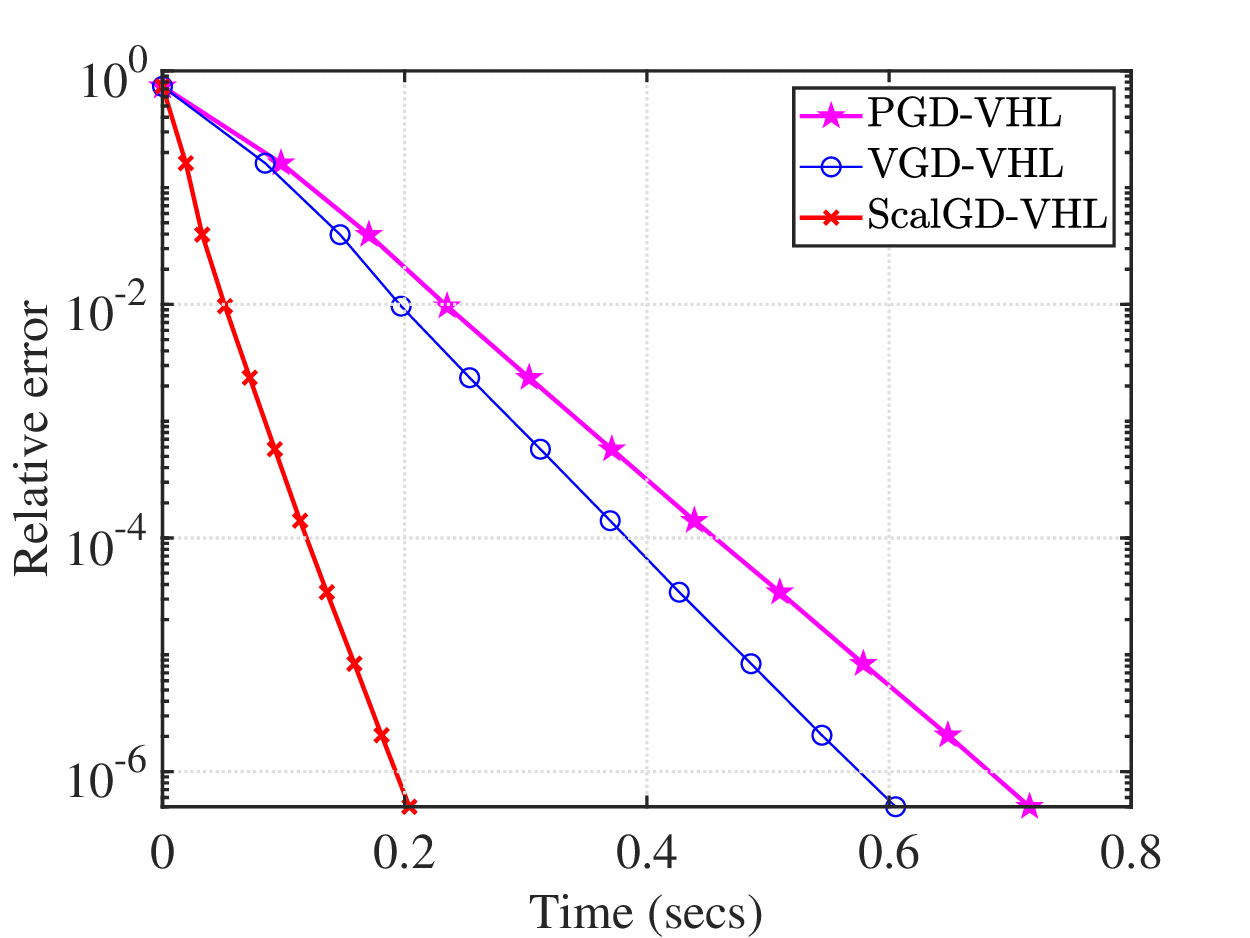}
    \caption{The average running time required for PGD-VHL, VGD-VHL, and ScalGD-VHL towards different recovery accuracies.}
    \label{fig:time_err}
\end{figure}

\subsection{ Running time comparisons}
\label{subsec:sim-runtime}
\begin{table*}[h!]	
		\begin{center}
			
			\caption{The average running time (in seconds) for PGD-VHL, VGD-VHL, and ScalGD-VHL to reach a fixed accuracy $10^{-7}$. }
			\begin{tabular}{c|c|c|c|c|c|c}
   \hline
				\textbf{$n$} & \multicolumn{3}{c|}{$1024$} &  \multicolumn{3}{c}{$2048$}    \\
	\hline			
				$(s,r)$ &{$(4,48)$} &{$(4,38)$} &{$(6,38)$} &{$(4,48)$} &{$(4,38)$} &{$(6,38)$} 
    \\ 
    \hline
				\hline
				PGD-VHL   &$24.8\pm8.3$ & $16.1\pm5.3$  & $25.8\pm9.0$  & $35.2\pm12.3$ & $29.1\pm13.1$  & $36.5\pm13.3$  \\
				VGD-VHL    &$21.7\pm7.3$  & $14.3\pm4.8$  & $22.7\pm8.1$  & $32.4\pm 11.3$  & $27.0\pm 12.3$  & $33.3\pm12.2$ \\
				ScalGD-VHL    &$\bm{3.3\pm 0.3}$  & $\bm{2.4\pm0.2}$ &  $\bm{5.0\pm0.4}$  & $\bm{5.3\pm0.1}$  & $\bm{4.3 \pm0.2}$  & $\bm{6.5\pm0.4}$ \\
    \hline
			\end{tabular}
			\label{tab:comput}
		\end{center}		
	\end{table*}
\begin{figure}[!t]
		\centering
	  \includegraphics[width=0.67\linewidth]{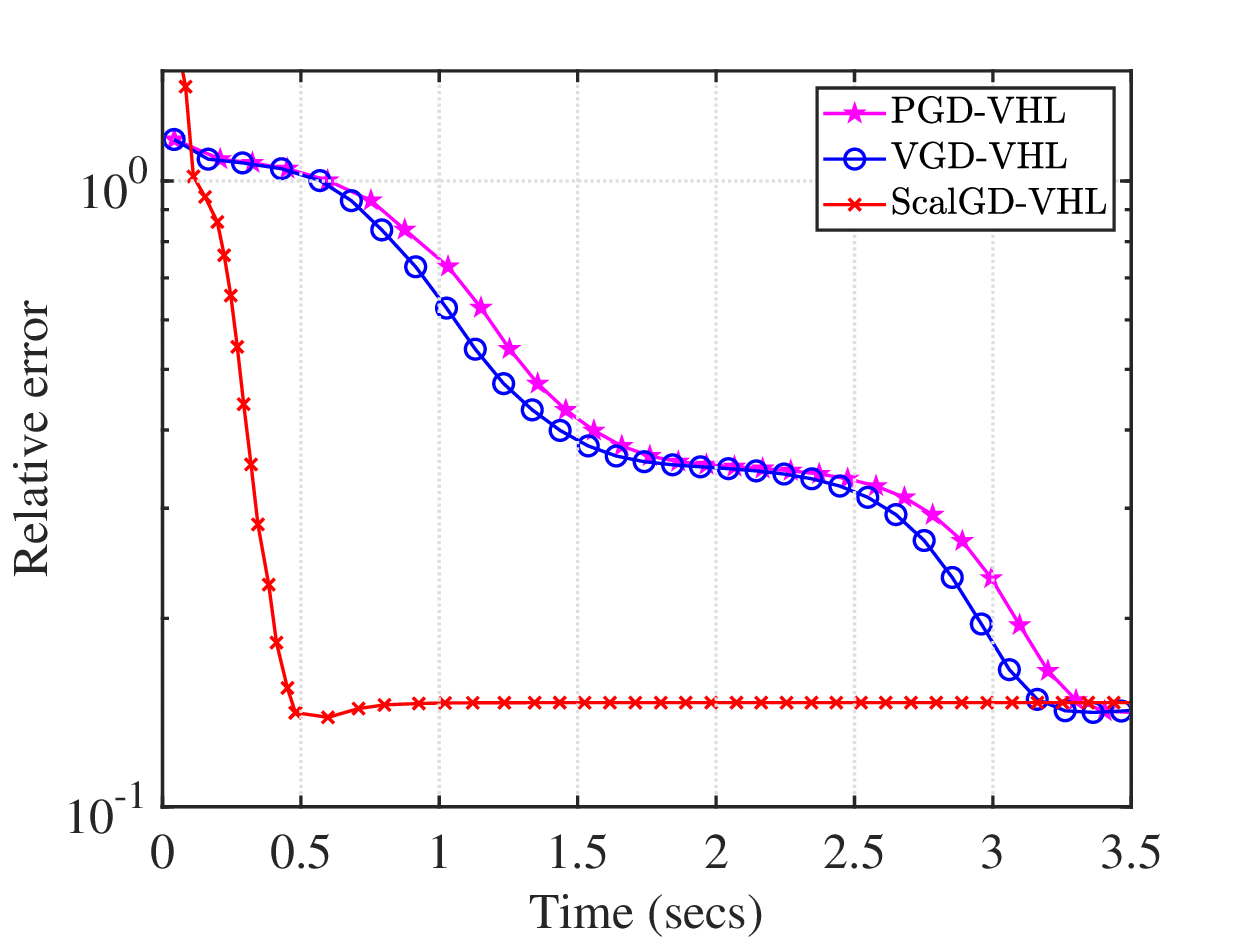}	
  
 {\caption{The recovery errors of PGD-VHL, VGD-VHL, and ScalGD-VHL in the signal domain.}
		\label{fig:converge_realcase_signal}}
\end{figure}
\begin{figure}[!t]
		\centering
	  \includegraphics[width=0.67\linewidth]{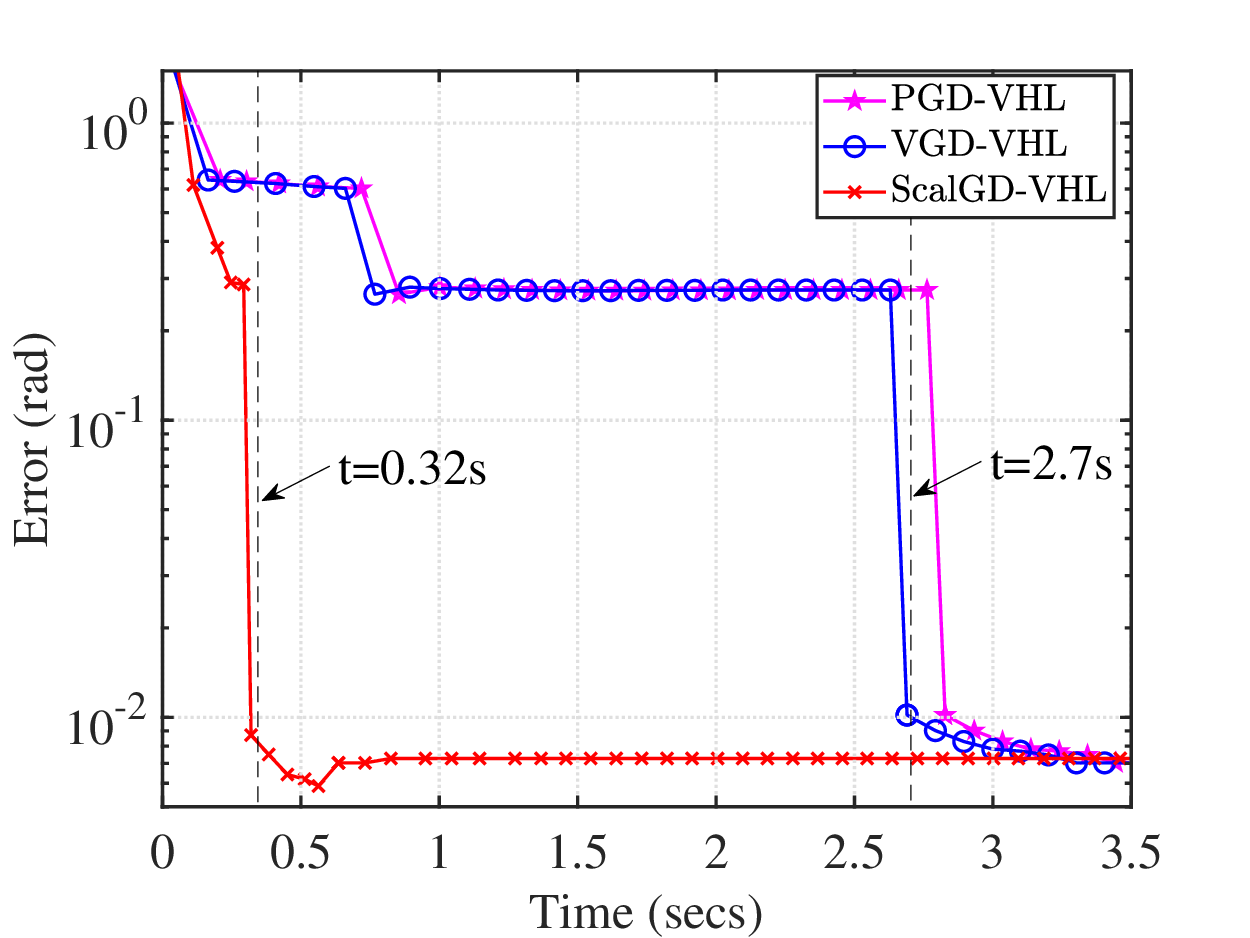}	
  
  {\caption{The recovery errors of PGD-VHL, VGD-VHL, and ScalGD-VHL in the delay-Doppler domain. }
		\label{fig:converge_realcase_freq}}
\end{figure}
\begin{figure}[!t]
		\centering	
 \subfigure[{Maximum runtime $t=0.32s$}]{
  \includegraphics[width=0.32\linewidth]{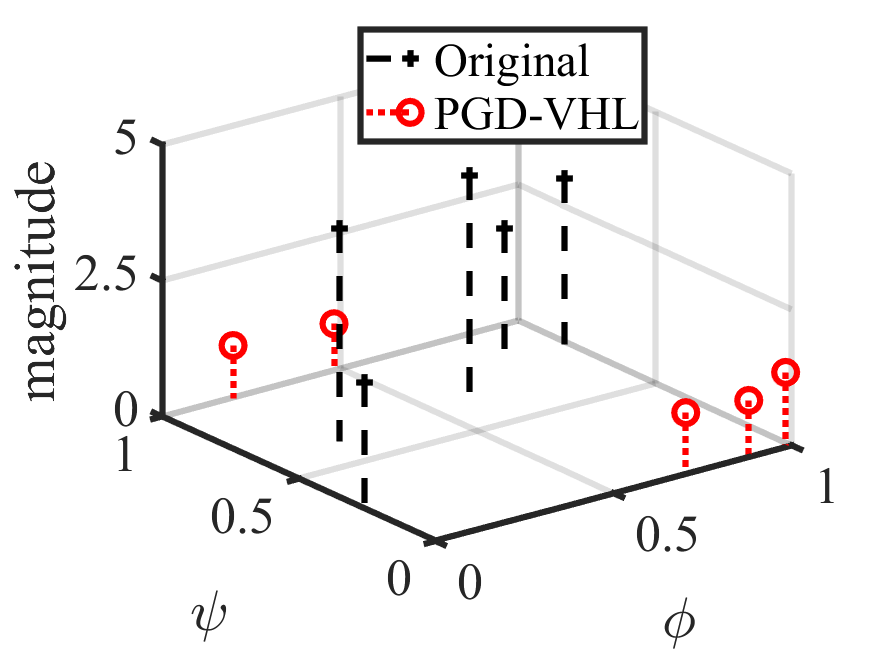} 
 \hfill
  \includegraphics[width=0.32\linewidth]{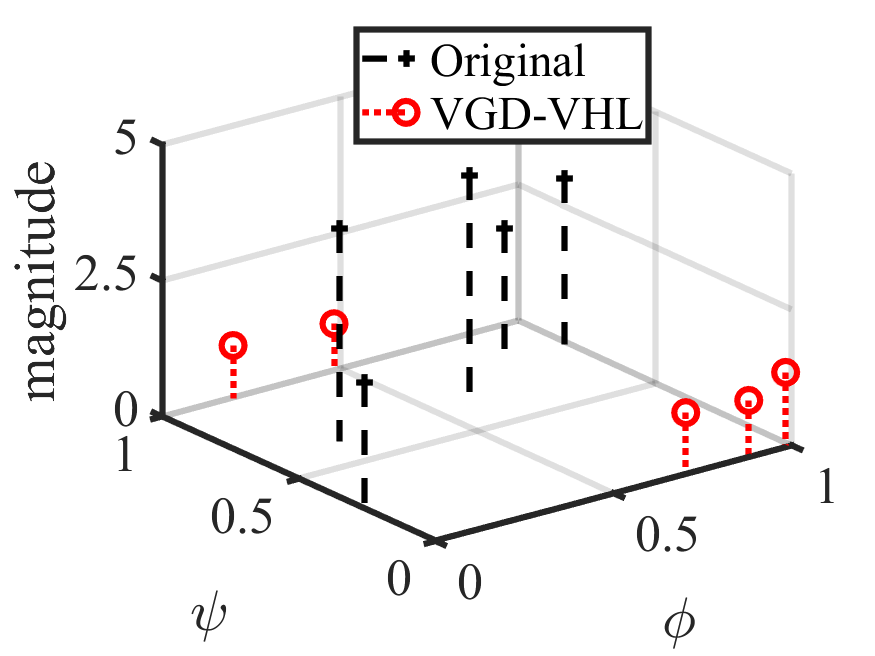} 
  \hfill
   \includegraphics[width=0.32\linewidth]{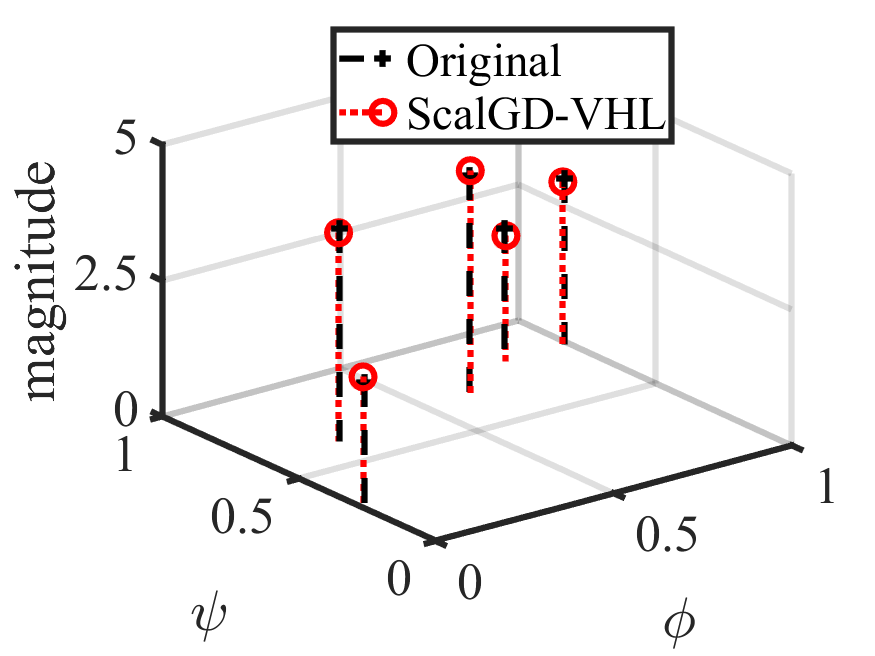}	
  }

  \centering	
  \subfigure[{Maximum runtime $t=2.7s$}]{
 \includegraphics[width=0.32\linewidth]{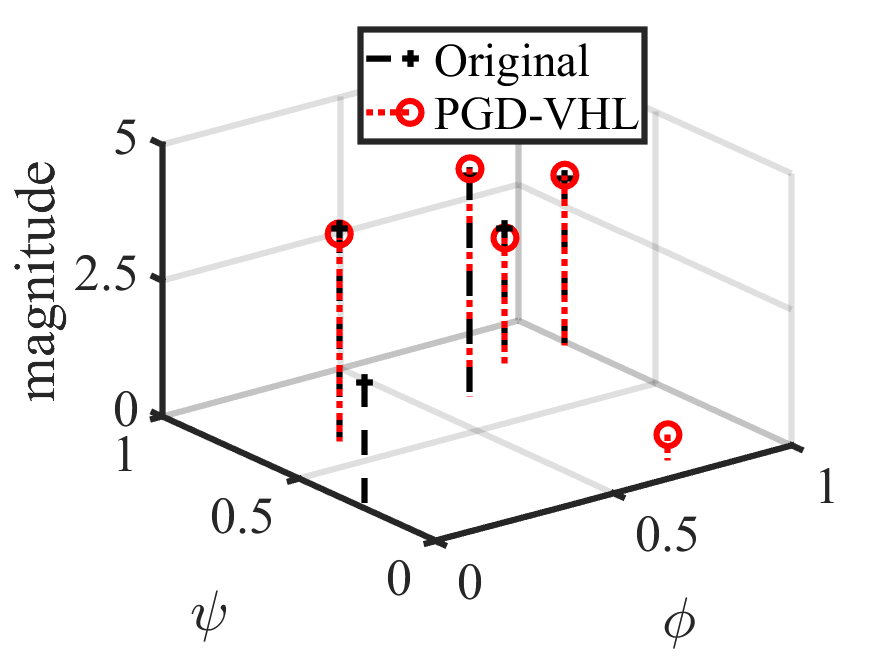} 
\hfill
  \includegraphics[width=0.32\linewidth]{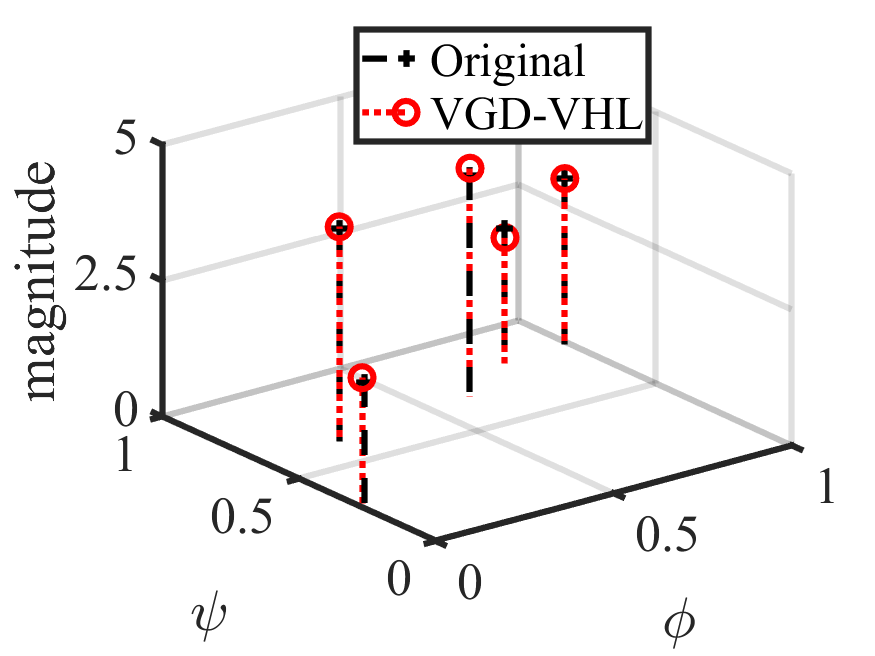} 
 \hfill
  \includegraphics[width=0.32\linewidth]{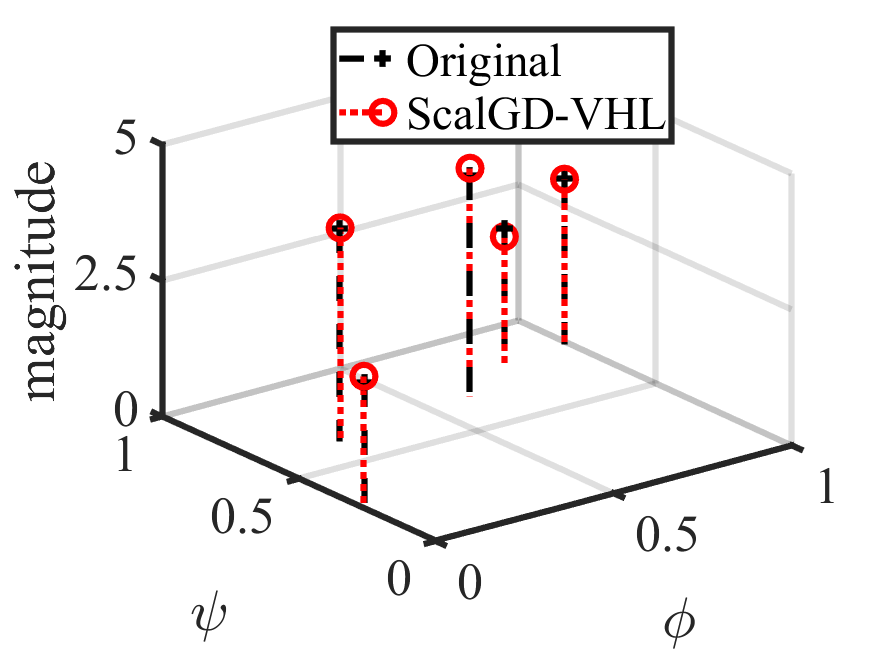}	
  } 
  \centering	
    {\caption{The recovery in the delay-Doppler domain for PGD-VHL, VGD-VHL, and ScalGD-VHL.} 
		\label{fig:2d_super_resolution}}
\end{figure}

In this subsection, we study the computational costs of our algorithms. The arguments are the same as in 
\ref{subsec:phasetran}.  Each experiment is repeated for 20 trials. First, we record the average running time required for PGD-VHL, VGD-VHL, and ScalGD-VHL to attain different recovery accuracies. We set $n=128$, $s=4$, and $r=8$ in this experiment. From Fig. \ref{fig:time_err}, ScalGD-VHL exhibits the lowest computation cost to reach the same recovery accuracy among the three algorithms. VGD-VHL converges faster than PGD-VHL in terms of runtime because VGD-VHL eliminates the computation of balancing terms and the projection. 

Secondly, we summarise three algorithms' average runtime required to reach a fixed accuracy as $\|\mX^k-\mX_\star\|_F/\|\mX_\star\|_F=10^{-7}$  
in Table \ref{tab:comput}. We consider two different lengths of the signal: $n=1024$ and $n=2048$ respectively, with three $(s,r)$ settings as $(4,48),(4,38)$ and $(6,38)$. From Table \ref{tab:comput}, we conclude that ScalGD-VHL exhibits the lowest computational costs and standard deviations compared to the other algorithms. VGD-VHL performs slightly better than PGD-VHL in terms of runtime and standard deviations.

  \subsection{Real 2D case} \label{subsec:real2D}
  We consider the joint delay-Doppler estimation problem. From \cite{Zheng2017}, the received signals are: 
  \begin{align*}
      \mY(m,n)=\sum_{k=1}^rd_ke^{-2\pi i(n\cdot\Delta f\tau_k+m\cdot Tf_k)}\mG(m,n),\numberthis\label{eq:2d-delay-doppler}
  \end{align*}
  where $(m,n)\in[M]\times[N]$, $M$ is the number of symbol blocks and $N$ is the number of orthogonal subcarriers. In the equation above, $\{d_k\}_{k=1}^r$ denote the channel coefficients, $\{\tau_k,f_k\}_{k=1}^r$ are the delays and Doppler frequencies, $\Delta f$ is the frequency spacing between adjacent subcarriers, $T$ is the symbol duration, and $\mG(m,n)$ refers to the symbol modulated on the $n$-th subcarrier and in the $m$-th block. Without loss of generality, we denote $\phi_k=\Delta f\tau_k\in[0,1)$ and $\psi_k=Tf_k\in[0,1)$. 
  
  Let $\vg=\mathrm{vec}(\mG)$, $\vy=\mathrm{vec}(\mY)$. \eqref{eq:2d-delay-doppler} can be rewritten as: $
      \vy =\(\sum_{k=1}^rd_k\va_{\phi_k}\otimes\va_{\psi_k}\)\odot\vg$. As noted in \cite{Mao2022}, $\vg$ can be approximately represented as $\vg=\mB\vh$. \eqref{eq:2d-delay-doppler} can be further reformulated  as 
 \begin{align*}
 	y_j = \la\vb_j\ve_j^T, \mX_\star \ra,\quad j=0,\cdots,MN-1,
 \end{align*}
 where $\mX_\star=\sum_{k=1}^r d_k\vh(\va_{\phi_k}\otimes\va_{\psi_k})^T$. We apply our approach to this 2-D blind super-resolution problem and set $M=17$, $N=13$, $(s,r)=(4,5)$ in experiments. 
 The row $\vb_k$ of $\mB$ is independently generated by $\vb_k=[1~e^{2\pi if_k} ... e^{2\pi i(s-1)f_k}]$ where $f_k$ is taken uniformly from $[0,1]$. Besides, 
 the measurements $\vy$ are corrupted by the random noise $\ve=\sigma_e\cdot\|\vy\|_2\cdot\vw/\|\vw\|_2$. Here, the elements of $\vw$ are i.i.d standard Gaussian variables and the noise level is $\sigma_e=0.15$. Other arguments are the same as in \ref{subsec:phasetran}. 
 
 Three {gradient} methods are applied to the real case, which are PGD-VHL, VGD-VHL, and ScalGD-VHL.  The delay-Doppler parameters $\{\phi_k,\psi_k\}_{k=1}^r$ can be retrieved through the 2D MUSIC \cite{Zheng2017} after obtaining the signal matrix $\mX$. The channel coefficients $\{d_k\}_{k=1}^r$ are estimated as in \cite{Chen2022}. {We compare the recovery errors of these methods as a function of time in the signal domain and delay-Doppler domain, respectively.  The recovery error in the signal domain is defined as $\|\mX-\mX_\star\|_F/\|\mX_\star\|_F$. The recovery error in the  delay-Doppler domain is defined as $\sqrt{\sum_{k=1}^r ((\hat{\phi}_k-\phi_k)^2+ 
   (\hat{\psi}_k-\psi_k)^2)}
$, where $(\hat{\phi}_k,\hat{\psi}_k)$ are the estimated delay-Doppler parameters and  $(\phi_k,\psi_k)$ are the true parameters. 
 As shown in Fig. \ref{fig:converge_realcase_signal} and Fig. \ref{fig:converge_realcase_freq}, ScalGD-VHL converges fastest, successfully recovering the delay-Doppler parameters at $t=0.32s$, 
 and VGD-VHL could successfully recover the delay-Doppler parameters at $t=2.7s$. Thus we select the maximum runtime as {$t=0.32s$} and $t=2.7s$ to present our advantages over PGD-VHL.} 

{We show the recovery of these methods in the delay-Doppler domain visually in Fig. \ref{fig:2d_super_resolution}. The results show that only ScalGD-VHL can successfully recover the true delay-Doppler parameters at $t=0.32s$. When the maximum runtime is $t=2.7s$,} both VGD-VHL and ScalGD-VHL can recover the parameters almost exactly, while PGD-VHL misses a target.

\section{Conclusions} \label{sec:conclusions}
Towards the problem of blind super-resolution of point sources, we propose a simpler unconstrained optimization problem without a balancing term and incoherence constraint via vectorized Hankel lift. We develop two simpler and provable gradient methods based on low-rank factorization, named VGD-VHL and ScalGD-VHL.~
Our methods enjoy lower theoretical iteration complexity compared to the prior method PGD-VHL. 
Additionally, ScalGD-VHL converges fastest whose iteration complexity is independent of $\kappa$. Numerical results verify that our methods exhibit superior computational efficiency and achieve comparable recovery performance to previous approaches. 

\appendices  
\section{Proof of Theorem \ref{thm:VGD-VHL}} \label{apd:pf-VGD}
{We show the linear convergence mechanism after the gradient updating in Appendix \ref{subsec:conver_VGD}.  
Then we finish the proof of Theorem \ref{thm:VGD-VHL} in Appendix \ref{subsec:formpf_VGD}. 
} First of all, we characterize the region of the local basin of attraction for VGD-VHL as:
\begin{align*}
    \B_{\mathrm{v}}(\delta_1,\delta_2)=&\{ \mF=   
    \begin{bmatrix}
        \mL^H\mR^H
    \end{bmatrix}^H|
    \|\mL\mR^H-\mM_\star\|_F\leq\delta_1\sigma_r(\mM_\star),\\
    &\|\mL^H\mL-\mR^H\mR\|_F\leq\delta_2\sigma_r(\mM_\star)\}. \numberthis \label{eq:basin_vgd}
\end{align*}

\subsection{{Linear convergence mechanism}} \label{subsec:conver_VGD}
{
 In what follows, the implicit balancing mechanism for VGD-VHL is shown in Lemma~\ref{lem:approx-balance}. To establish the linear convergence, 
   it is important to demonstrate PL inequality and the smoothness condition for $f_0(\mF)$, seeing Lemma~\ref{lem:pl-fro}  and Lemma~\ref{lem:smoothness-fro}.  
However, the gradient updating direction is $\nabla f(\mF)=\nabla f_0(\mF)+\nabla f_1(\mF)$, not $\nabla f_0(\mF)$. Fortunately, we can obtain the intuition that $\nabla f(\mF)\approx\nabla f_0(\mF)$ by establishing $\|\nabla f_1(\mF)\|_F\leq\delta \|\nabla f_0(\mF)\|_F$ in Lemma~\ref{lem:upbd-perturbgd-fro} via a less incoherence-demanding analysis, where $\delta$ is a small constant. Then the linear convergence result can be established in Lemma~\ref{lem:convergence-vanilla}.}

We introduce the approximate balancing of factors after gradient updating, ensuring the iterates lie in the local region. 
 
 \begin{lemma}[Approximate balancing]\label{lem:approx-balance}
     For the gradient updating $\mF_{+}=\mF-\eta\nabla f(\mF)$ {where $\mF_{+}=\begin{bmatrix}
 	\mL_{+}^H~\mR_{+}^H
 \end{bmatrix}^H$}, one has
     \begin{align*}
         \|\mL_{+}^H\mL_{+}-\mR_{+}^H\mR_{+}\|_F\leq\|\mL^H\mL-\mR^H\mR\|_F+\eta^2\|\nabla f(\mF)\|_F^2. 
     \end{align*}
 \end{lemma}
 \begin{proof}        
We set $\mE=(\G(\A^{*}\A-\I)\G^{*}+\I)(\mL\mR^H-\mM_\star)$, and it is obvious that $\nabla_{\mL} f(\mF)=\mE\mR$ and $\nabla_{\mR} f(\mF)=\mE^H\mL$.
    From $\mL_{+}=\mL-\eta\nabla_{\mL} f(\mF)=\mL-\eta\mE\mR$ and $\mR_{+}=\mR-\eta\nabla_{\mR} f(\mF)=\mR-\eta\mE^H\mL$, one can obtain:
    \begin{align*}
        &\|\mL_{+}^H\mL_{+}-\mR_{+}^H\mR_{+}\|_F\\
        &{=}\|\mL^H\mL-\mR^H\mR+\eta^2(\mE\mR)^H(\mE\mR)-\eta^2(\mE^H\mL)^H(\mE^H\mL)\|_F\\
        &{\leq} \|\mL^H\mL-\mR^H\mR\|_F+\eta^2(\|\mE\mR\|_F^2+\|\mE^H\mL\|_F^2)
       \\&{=} \|\mL^H\mL-\mR^H\mR\|_F+\eta^2\|\nabla f(\mF)\|_F^2,
    \end{align*}
where some cross-terms are canceled in the second line.
 \end{proof}
 When  $\|\mL^H\mL-\mR^H\mR\|_F$ and $\eta^2\|\nabla f(\mF)\|_F^2$ are small during the trajectories, $ \|\mL_{+}^H\mL_{+}-\mR_{+}^H\mR_{+}\|_F$ is upper bounded, which means it is kept approximately balanced. 
 
 Then we introduce the PL inequality {in the local region} and 
 denote $f_0^\star$ as the minimum value of $f_0(\mF)$, which is zero. 

\begin{lemma}[PL inequality]\label{lem:pl-fro}
Provided $\mF\in\B_{\mathrm{v}}(\frac{1}{64},\frac{1}{20}),$ one has 
    \begin{align*}
   	\norm{\nabla f_0(\mF)}_F^2\geq \frac{\sigma_r(\mM_\star)}{7}\left(f_0(\mF)-f_0^\star\right).
     \end{align*}
\end{lemma}

\begin{proof}
See Appendix~\ref{apd:pl-fro}.
\end{proof}

 { The gradient updating direction is not $\nabla f_0(\mF)$, but we can obtain the intuition that $\nabla f(\mF)\approx\nabla f_0(\mF)$ by establishing $\|\nabla f_1(\mF)\|_F\leq\delta \|\nabla f_0(\mF)\|_F$ in Lemma~\ref{lem:upbd-perturbgd-fro}, via our less incoherence-demanding analysis.}

\begin{lemma}[Upper bound of $\|\nabla f_1(\mF)\|_F$]\label{lem:upbd-perturbgd-fro}
 Provided $n\gtrsim O(\delta_0 \mu_0\mu_1sr\log(sn))$
 ~and $\mF\in\B_{\mathrm{v}}(\frac{\delta_0}{20\sqrt{\kappa \mu_0 s}},\frac{1}{20})$ where $\delta_0\leq\frac{1}{4},$ one has
\begin{align*}
    \|\nabla f_1(\mF)\|_F\leq&\frac{25}{4}\delta_0 \norm{\nabla f_0(\mF)}_F,
\end{align*}
    with probability at least $1-(sn)^{-c}$.
\end{lemma}
\begin{proof}
   See Appendix~\ref{apd:upbd-grad-fro}.
\end{proof}

 Lemma~\ref{lem:upbd-perturbgd-fro} tells that the perturbation from the sensing operator 
 in blind super-resolution is sharply small and {thus 
 $\mF_{+}\approx \mF-\eta\nabla f_0(\mF)$.} {The linear convergence towards $f_0(\mF)$ can be derived as in \cite{karimi2016}, by PL inequality Lemma~\ref{lem:pl-fro} and the smoothness condition Lemma~\ref{lem:smoothness-fro}  in what follows.} 

 \begin{lemma}[Local smoothness
 ]\label{lem:smoothness-fro}
Let $\mF_{+}=\mF-\eta\nabla f(\mF)$ for $\eta\leq 1/\sigma_1(\mM_\star)$.  
When $\mF\in\B_{\mathrm{v}}(\frac{1}{20},\frac{1}{20})$ and $
    \|\nabla f_1(\mF)\|_F\leq \varepsilon \norm{\nabla f_0(\mF)}_F$, one has 
   \begin{align*}
		f_0(\mF_+)\leq f_0(\mF)+\Real\ip{\nabla f_0(\mF)}{\mF_+-\mF}+\frac{l}{2}\norm{\mF_{+}-\mF}_F^2,
	\end{align*} 
 where $l=(4+\frac{\varepsilon^2}{80})\sigma_1(\mM_\star)$. 
\end{lemma} 
\begin{proof}
    See Appendix \ref{apd:smooth-fro}. 
\end{proof}

{Finally, we show the linear convergence of $f_0(\mF_{+})-f_0^\star$.}
\begin{lemma}[Linear convergence of VGD-VHL]\label{lem:convergence-vanilla}
Let $\mF_{+}=\mF-\eta\nabla f(\mF)$ for $\eta\leq \frac{1}{25\sigma_1(\mM_\star)}$. When $\mF\in\B_\mathrm{v}(\frac{\delta_0}{20\sqrt{\kappa\mu_0 s}},\frac{1}{20})$ and $n\gtrsim O(\delta_0 \mu_0\mu_1sr\log(sn))$ for $\delta_0\leq\frac{2}{25}$, one has  
\begin{align*}
   f_0(\mF_{+})-f_0^\star\leq(1-{\eta\sigma_r(\mM_\star)}/{28})(f_0(\mF)-f_0^\star),
\end{align*}
    with probability at least $1-(sn)^{-c}$. 
\end{lemma}
 \begin{proof}
 {
Firstly, we check the conditions and constant in Lemma~\ref{lem:smoothness-fro} and apply it. From Lemma~\ref{lem:upbd-perturbgd-fro}, one can obtain $\|\nabla f_1(\mF)\|_{F}\leq \frac{25}{4}\delta_0 \norm{\nabla f_0(\mF)}_{F}\leq\frac{1}{2}\norm{\nabla f_0(\mF)}_{F}$. }

Then $\varepsilon=\frac{1}{2}$ and $l\leq\frac{9}{2}\sigma_1(\mM_\star)$ in Lemma~\ref{lem:smoothness-fro}, thus
\begin{align*}
      &f_0(\mF_{+})-f_0(\mF)\\
      &\leq\Real\langle\nabla f_0(\mF),-\eta\nabla f(\mF)\rangle{+}\frac{9}{4}\eta^2\sigma_1(\mM_\star)\norm{\nabla f(\mF)}_F^2
      \\&{\overset{(a)}{\leq}}-\eta\|\nabla f_0(\mF)\|_F^2-\eta\Real\langle\nabla f_0(\mF),\nabla f_1(\mF)\rangle
      \\&\quad+\frac{9}{2}\eta^2\sigma_1(\mM_\star)(\|\nabla f_0(\mF)\|_F^2+\|f_1(\mF)\|_F^2)
      \\&{\overset{(b)}{\leq}}-\eta(\frac{3}{4}-\frac{1}{5})\|\nabla f_0(\mF)\|_F^2+\eta(1+\frac{1}{5})\|\nabla f_1(\mF)\|_F^2
      \\&{\overset{(c)}{\leq}}-\frac{\eta}{4}\|\nabla f_0(\mF)\|_F^2,
  \end{align*}
  where {we split $\nabla f(\mF)$ as $\nabla f_0(\mF)+\nabla f_1(\mF)$ and invoke $(a+b)^2\leq2a^2+2b^2$ in (a)}, and {$\Real\langle \va,\vb\rangle=\|\frac{1}{2}\va+\vb\|_2^2-\frac{1}{4}\|\va\|_2^2-\|\vb\|_2^2\geq -\frac{1}{4}\|\va\|_2^2-\|\vb\|_2^2$, $9\eta\sigma_1(\mM_\star)/2\leq 1/5$ in (b)}, and {(c) results from Lemma~\ref{lem:upbd-perturbgd-fro}}.  

 
 By $\norm{\nabla f_0(\mF)}_F^2\geq \frac{\sigma_r(\mM_\star)}{7}\left(f_0(\mF)-f_0^\star\right)$ in Lemma~\ref{lem:pl-fro}, we can derive that
  \begin{align*}
      f_0(\mF_{+})-f_0^\star\leq(1-{\eta\sigma_r(\mM_\star)}/{28})(f_0(\mF)-f_0^\star).
\end{align*}
 \end{proof}
 \subsection{{Proof of Theorem \ref{thm:VGD-VHL}}} \label{subsec:formpf_VGD}
 
 {We first introduce the initialization condition:
\begin{lemma}[Initialization]\label{lem:spec-init}
    	Given the SVD of $\T_r(\G\A^{*}(\vy))$ being $\vU^0\vSS^0(\vV^0)^H$, the initialization is set as $\mL^0=\vU^0(\vSS^0)^{\frac{1}{2}}$,  $\mR^0=\vV^0(\vSS^0)^{\frac{1}{2}}$.   Under Assumption \ref{assumption1}, Assumption \ref{assumption2} and $n\geq C\varepsilon^{-2}\kappa^2\mu_0^2\mu_1 s^2r^2\log^2(sn)$, we have
     \begin{align*}
         \|\mL^0(\mR^0)^H-\mM_\star\|_F\leq\varepsilon\sigma_r(\mM_\star)/\sqrt{\mu_0 s}
     \end{align*}
     with probability at least $1-(sn)^{-c_1}$ for  
     constants $c_1,C>0$.
\end{lemma}
\begin{proof}
   By \cite[Lemma~\uppercase\expandafter{\romannumeral6}.3]{Mao2022}, with probability at least $1-(sn)^{-c_1}$, one has
     \begin{align*}
          \|\mL^0(\mR^0)^H-\mM_\star\|_F\leq \sqrt{2r}\|\mL^0(\mR^0)^H-\mM_\star\|\leq\frac{\varepsilon\sigma_r(\mM_\star)}{\sqrt{\mu_0 s}},
     \end{align*}
      provided $n\geq C\varepsilon^{-2}\kappa^2\mu_0^2\mu_1 s^2r^2\log^2(sn)$.
\end{proof}}
By \eqref{eq:estimate_upbd}, one has
$$\|\mX^k-\mX_\star\|_F^2 \leq\|\mL^k(\mR^k)^H-\mM_\star\|_F^2=2(f_0(\mF^k)-f_0^\star).$$
The key is to show the linear convergence of $f_0(\mF^k)-f_0^\star$, with the help of the balancing property of factors.
 We make the following induction hypotheses: \begin{align*}
     f_0(\mF^{k})-f_0^\star\leq\frac{1}{2}\Big(1-\frac{\eta\sigma_r(\mM_\star)}{28}\Big)^{k}(\frac{\delta_0\sigma_r(\mM_\star)}{20\sqrt{\kappa\mu_0 s}})^2,\numberthis\label{eq:induction-a}
 \end{align*}
as well as the approximate balancing of factors $\mL^k$ and $\mR^k$:
\begin{align*}
    \|(\mL^{k})^H\mL^k{-}(\mR^{k})^H\mR^k\|_F{\leq}\big(1-(1-\frac{\eta\sigma_r(\mM_\star)}{28})^{k}\big)\frac{\sigma_r(\mM_\star)}{20},\numberthis\label{eq:induction-b}
\end{align*}
we establish \eqref{eq:induction-a} and \eqref{eq:induction-b} via an inductive manner.

{{For step $k=0$}:} setting $\varepsilon=\frac{\delta_0}{20\sqrt{\kappa}}$ in Lemma~\ref{lem:spec-init}, we have 
 $$f_0(\mF^{0})-f_0^\star=\frac{1}{2}\|\mL^0(\mR^0)^H-\mM_\star\|_F^2\leq\frac{1}{2}(\frac{\delta_0\sigma_r(\mM_\star)}{20\sqrt{\kappa\mu_0 s}})^2,$$ and \eqref{eq:induction-a} is verified. {Besides, \eqref{eq:induction-b} is satisfied as $\|(\mL^0)^H\mL^0{-}(\mR^0)^H\mR^0\|_F{=}0$, which is balanced.}


{From the $k$-th to the $(k+1)$-th step:}
    {from \eqref{eq:induction-a} and \eqref{eq:induction-b}, we conclude that $\mF^k\in\B_\mathrm{v}(\frac{\delta_0}{20\sqrt{\kappa\mu_0 s}},\frac{1}{20})$. Then from Lemma~\ref{lem:convergence-vanilla}, we have}
\begin{align*}
   f_0(\mF^{k+1})-f_0^\star\leq(1-\frac{\eta\sigma_r(\mM_\star)}{28})(f_0(\mF^k)-f_0^\star),
\end{align*}
  thus \eqref{eq:induction-a} is established for the $(k+1)$-th step.
  
  Towards approximately balancing \eqref{eq:induction-b} for the $(k+1)$-th step, we set $c(k)=(1-\frac{\eta\sigma_r(\mM_\star)}{28})^k$ to simplify notations. {Combining Lemma~\ref{lem:approx-balance}, it suffices to show 
\begin{align*}
     \eta^2\|\nabla f(\mF^k)\|_F^2{\leq} \(c(k){-}c(k+1)\)\frac{\sigma_r(\mM_\star)}{20},
\end{align*}
provided $\|(\mL^{k})^H\mL^k{-}(\mR^{k})^H\mR^k\|_F\leq\(1-c(k)\)\frac{\sigma_r(\mM_\star)}{20}$.}

By $\nabla f(\mF^k)=\nabla f_0(\mF^k)+\nabla f_1(\mF^k)$, we establish
{
\begin{align*}
    \eta^2\|\nabla f(\mF^k)\|_F^2&\leq 2\eta^2(\|\nabla f_0(\mF^k)\|_F^2+\|\nabla f_1(\mF^k)\|_F^2)
    \\&\leq 3\eta^2\|\nabla f_0(\mF^k)\|_F^2
    \\&\leq 6\eta^2(\|\mL^k\|^2\vee\|\mR^k\|^2)\|\mL^k(\mR^k)^H-\mM_\star\|_F^2
    \\&\leq\frac{\eta c(k)\sigma_r^2(\mM_\star)}{560}{=}\(c(k){-}c(k+1)\)\frac{\sigma_r(\mM_\star)}{20},
\end{align*}
where the last inequality results from the facts $\|\mL^k\|^2\vee\|\mR^k\|^2\leq\frac{5}{4}\sigma_1(\mM_\star)$ from Lemma~\ref{lem:wellcond} when $\mF^k\in\B_{\mathrm{v}}(\frac{1}{20},\frac{1}{20})$, $\eta\leq\frac{1}{25\sigma_1(\mM_\star
)}$, and \eqref{eq:induction-a}.} Therefore, we finish the induction. 

\section{Proof of Theorem \ref{thm:ScalGD-VHL}} \label{apd:pf-ScaleGD}
{We show the linear convergence mechanism after the gradient updating in \ref{subsec:linconver_scal}.  
Then we finish the proof of Theorem \ref{thm:ScalGD-VHL} in \ref{subsec:formpf_scal}.} 
We first characterize the local basin of attraction towards ScalGD-VHL as:
\begin{align*}
    \B_{\mathrm{s}}(\delta){=}&\{ \mF{=}   
    \begin{bmatrix}
        \mL^H\mR^H
    \end{bmatrix}^H|
    \|\mL\mR^H{-}\mM_\star\|_F{\leq}\delta\sigma_r(\mM_\star)\}.\numberthis \label{eq:local_region_scal}
\end{align*}

\subsection{{Linear convergence mechanism}} \label{subsec:linconver_scal}
{
 In what follows,
   it is important to demonstrate PL inequality and the smoothness condition for $f_0(\mF)$ under the scaled norm, seeing Lemma~\ref{lem:pl-scale}  and Lemma~\ref{lem:smoothness-scale}. Also, as explained in Appendix \ref{apd:pf-VGD}, 
we need to show $\nabla f(\mF)\approx\nabla f_0(\mF)$ by establishing $\|\nabla f_1(\mF)\|_{\mL,\mR}^{*}\leq \delta \norm{\nabla f_0(\mF)}_{\mL,\mR}^{*}$ in Lemma~\ref{lem:upbd-gd-scale} via our tailored less incoherence-demanding analysis, where $\delta$ is a small constant. Then the linear convergence with the help of the scaled norm is shown in Lemma~\ref{lem:convergence-vanilla}. 
}

We introduce the PL inequality under the scaled norm in the local region. 
\begin{lemma}[PL inequality in the {dual} scaled norm]\label{lem:pl-scale}
When $\mF\in\B_\mathrm{s}(\frac{1}{64}),$ one has 
    \begin{align*}
   	(\norm{\nabla f_0(\mF)}_{\mL,\mR}^*)^2\geq \frac{9}{50}\left(f_0(\mF)-f_0^\star\right).
     \end{align*}
\end{lemma}
\begin{proof}
See Appendix~\ref{apd:pl-scale}.
\end{proof}
It is necessary to show $\|\nabla f_1(\mF)\|_F\leq\delta \|\nabla f_0(\mF)\|_F$ for VGD-VHL as pointed out in Appendix \ref{apd:pf-VGD}. Similarly, we need to show $\|\nabla f_1(\mF)\|_{\mL,\mR}^{*}\leq \delta \norm{\nabla f_0(\mF)}_{\mL,\mR}^{*}$ via a tailored less incoherence-demanding analysis for 
ScalGD-VHL, which is defined in the dual scaled norm. 
\begin{lemma}[Upper bound of $\|\nabla f_1(\mF)\|_{\mL,\mR}^{*}$]\label{lem:upbd-gd-scale}When $n\gtrsim O(\delta \mu_0\mu_1sr\log(sn))$ and $\mF\in\B_\mathrm{s}(\frac{\delta}{20\sqrt{\mu_0 s}})$ for $\delta\leq\frac{1}{4}:$  
\begin{align*}
    \|\nabla f_1(\mF)\|_{\mL,\mR}^{*}\leq 5\delta \norm{\nabla f_0(\mF)}_{\mL,\mR}^{*}
\end{align*}
   holds with probability at least $1-(sn)^{-c}$. 
\end{lemma}
\begin{proof}
    See Appendix~\ref{apd:upbd-grad-scale}.
\end{proof}

  Then we introduce the smoothness condition in the scaled norm under the gradient updating direction in ScalGD-VHL. {We denote $\mF_{+}$ for the next iterate as: 
  \begin{align*}
      \mF_{+}=\mF-\eta\begin{bmatrix}
    \nabla_{\mL} f(\mF)(\mR^H\mR)^{-1}\\
     \nabla_{\mR} f(\mF)(\mL^H\mL)^{-1} 
\end{bmatrix}.\numberthis \label{eq:nextit_scal}
  \end{align*}}
\begin{lemma}[Local smoothness of $f_0$ in the scaled norm]\label{lem:smoothness-scale}
 Provided $\mF\in\B_\mathrm{s}(\frac{1}{11})$ and $
    \|\nabla f_1(\mF)\|_{\mL,\mR}^{*}\leq \varepsilon \norm{\nabla f_0(\mF)}_F$.  {For $\mF_{+}$ in \eqref{eq:nextit_scal} where $\eta\leq 1$, one has} 
  \begin{align*}
		f_0(\mF_+)\leq f_0(\mF)+\Real\ip{\nabla f_0(\mF)}{\mF_+-\mF}+\frac{l}{2}\norm{\mF_{+}-\mF}_{\mL,\mR}^2,
	\end{align*} 
 where $l=(78+\varepsilon^2)/25.$
 
\end{lemma} 
\begin{proof}
    See Appendix~\ref{apd:smooth-scale}.
\end{proof}
Equipped with the previous necessary ingredients, we show the linear convergence result for $f_0(\mF_{+})-f_0^\star$.
\begin{lemma}[Linear convergence of ScalGD-VHL]\label{lem:convergence-scale}
 Provided $n\gtrsim O(\delta_0 \mu_0\mu_1sr\log(sn))$ and $\mF\in\B_\mathrm{s}(\frac{\delta_0}{20\sqrt{\mu_0 s}})$ where $\delta_0\leq\frac{1}{10}$.{ For $\mF_{+}$ in \eqref{eq:nextit_scal} where $\eta\leq \frac{1}{4}$, one has } 
\begin{align*}
   f_0(\mF_{+})-f_0^\star\leq\(1-\frac{\eta}{25}\)(f_0(\mF)-f_0^\star),
\end{align*}
    with probability at least $1-(sn)^{-c}$. 
\end{lemma}
\begin{proof}
{
Firstly, we check the conditions and constant in Lemma~\ref{lem:smoothness-scale} and apply it. From Lemma~\ref{lem:upbd-gd-scale}, one can obtain $\|\nabla f_1(\mF)\|_{\mL,\mR}^{*}\leq 5\delta \norm{\nabla f_0(\mF)}_{\mL,\mR}^{*}\leq\frac{1}{2}\norm{\nabla f_0(\mF)}_{F}$. 

Then $\varepsilon=\frac{1}{2}$ and $l\leq\frac{16}{5}$ in Lemma~\ref{lem:smoothness-scale}, thus}
  \begin{align*}
&f_0(\mF_{+}){-}f_0(\mF)\leq \Real\left\langle\small{\begin{bmatrix}
    \nabla_{\mL} f_0(\mF)\\
    \nabla_{\mR} f_0(\mF)
\end{bmatrix}},-\eta\small{\begin{bmatrix}
    \nabla_{\mL} f(\mF)(\mR^H\mR)^{-1} \\
    \nabla_{\mR} f(\mF)(\mL^H\mL)^{-1}
\end{bmatrix}}\right\rangle\\
&\qquad +\frac{8\eta^2}{5}\norm{\small{\begin{bmatrix}
    \nabla_{\mL} f(\mF)\big((\mR)^H\mR\big)^{-1} \\
    \nabla_{\mR} f(\mF)\big((\mL)^H\mL\big)^{-1}
\end{bmatrix}}}_{\mL,\mR}^2\\
&{\overset{(a)}{=}}\frac{8\eta^2}{5}(\|\nabla f_0(\mF)+\nabla f_1(\mF)\|_{\mL,\mR}^{*})^2-\eta(\|\nabla f_0(\mF)\|_{\mL,\mR}^{*})^2\\&\quad-\eta\Real\left\langle\small{\begin{bmatrix}
    \nabla_{\mL} f_0(\mF)(\mR^H\mR)^{-\frac{1}{2}}\\
    \nabla_{\mR} f_0(\mF)(\mL^H\mL)^{-\frac{1}{2}}
\end{bmatrix}},\small{\begin{bmatrix}
    \nabla_{\mL} f_1(\mF)(\mR^H\mR)^{-\frac{1}{2}} \\
    \nabla_{\mR} f_1(\mF)(\mL^H\mL)^{-\frac{1}{2}}
\end{bmatrix}}\right\rangle
\\&{\overset{(b)}{\leq}} \frac{16\eta^2}{5}\((\|\nabla f_0(\mF)\|_{\mL,\mR}^{*})^2+(\|\nabla f_1(\mF)\|_{\mL,\mR}^{*})^2\)\\
&\quad-\frac{3}{4}\eta(\|\nabla f_0(\mF)\|_{\mL,\mR}^{*})^2+\eta(\|\nabla f_1(\mF)\|_{\mL,\mR}^{*})^2
\\&{\overset{(c)}{\leq}} -\frac{1}{4}\eta(\|\nabla f_0(\mF)\|_{\mL,\mR}^{*})^2,
\end{align*}
where {in (a) we split $\nabla f(\mF)$ as $\nabla f_0(\mF)+\nabla f_1(\mF)$}, and {in (b) we apply $(a+b)^2\leq 2a^2+2b^2$, $\Real\langle \va,\vb\rangle\geq -\frac{1}{4}\|\va\|_2^2-\|\vb\|_2^2$}, and {(c) results from $\|\nabla f_1(\mF)\|_{\mL,\mR}^{*}\leq\frac{1}{2}\norm{\nabla f_0(\mF)}_{\mL,\mR}^{*}$ as well as $\eta\leq\frac{1}{4}$}.

By Lemma~\ref{lem:pl-scale}, we know$	(\norm{\nabla f_0(\mF)}_{\mL,\mR}^*)^2\geq\frac{4}{25}\left(f_0(\mF)-f_0^\star\right)$ and then we can derive that
\begin{align*}
      f_0(\mF_{+})-f_0^\star\leq\(1-\frac{\eta}{25}\)(f_0(\mF)-f_0^\star).
\end{align*}

\end{proof}
\subsection{{Proof of Theorem \ref{thm:ScalGD-VHL}}} \label{subsec:formpf_scal}
As in Appendix \ref{apd:pf-VGD}, the key is to establish the linear convergence of $f_0(\mF^k)-f_0^\star$.
  In what follows, we prove that 
 \begin{align*}
      f_0(\mF^{k})-f_0^\star\leq\frac{1}{2}\(1-\frac{\eta}{25}\)^{k}(\frac{\delta_0\sigma_r(\mM_\star)}{20\sqrt{\mu_0 s}})^2.\numberthis\label{eq:induction-scale}
       \end{align*}

{For step $k=0$:} {setting $\varepsilon=\frac{\delta_0}{20}$ in Lemma~\ref{lem:spec-init}, we have }
\begin{align*}
   f_0(\mF^{0})-f_0^\star=\frac{1}{2}\|\mL^0(\mR^0)^H-\mM_\star\|_F^2\leq\frac{1}{2}(\frac{\delta_0\sigma_r(\mM_\star)}{20\sqrt{\mu_0 s}})^2.
\end{align*}

{From the $k$-th to the $(k+1)$-th step:} from \eqref{eq:induction-scale}, we conclude that 
$\mF^k\in\B_s(\frac{\delta_0}{20\sqrt{\mu_0 s}})$, lying in a local basin of attraction.
Invoking Lemma~\ref{lem:convergence-scale} to establish 
\begin{align*}
      f_0(\mF^{k+1})-f_0^\star\leq\(1-\frac{\eta}{25}\)(f_0(\mF^k)-f_0^\star),
\end{align*}
then \eqref{eq:induction-scale} is established for the next step. Therefore, we finish the induction. 

\section{Proof of Polyak-Łojasiewicz inequalities} \label{apd:C-pl}
In the following Appendices \ref{apd:C-pl}, \ref{apd:proof-smooth}, \ref{apd:E-upbd-grad}, and \ref{apd:F-support-lems}, we apply the following notations throughout.

\subsection{Additional notations\label{addNotations}}
{~}{Define a function $\mathrm{sgn}(\cdot)$ as $\mathrm{sgn}(\mM)=\mU\mV^H$ when the SVD of $\mM$ being $\mU\vSS\mV^H$.}

For $\mM_\star=\H(\mX_\star)=\mU_\star\vSS_\star{\mV_\star}^H$ in Assumption \ref{assumption1}, let $\mL_\star=\mU_\star\vSS_\star^{\frac{1}{2}}$, $\mR_\star=\mV_\star\vSS_\star^{\frac{1}{2}}$ and we denote $\mF_\star=\begin{bmatrix}
      \mL_\star^H~
      \mR_\star^H
  \end{bmatrix}^H.$ 

  For a rank-$r$ matrix $\mM$ whose compact SVD is $\mU\vSS\mV^H$, let $\Tilde{\mL}=\mU\vSS^\frac{1}{2}, \Tilde{\mR}=\mV\vSS^\frac{1}{2}$, and we denote $\Tilde{\mF}=\begin{bmatrix}
        \Tilde{\mL}^H~\Tilde{\mR}^H
\end{bmatrix}^H.$ 

Denote $\mQ^\natural\coloneqq \mathrm{sgn}(\Tilde{\mF}^H\mF_\star)$, and {without special instructions, we denote $\mL=\Tilde{\mL}\mQ^\natural$, $\mR=\Tilde{\mR}\mQ^\natural$. Define} 
\begin{align*}
    \vDeltal\coloneqq\mL-\mL_\star,~~ \vDeltar\coloneqq\mR-\mR_\star,
\end{align*}
and we have the following decompositions:
    \begin{align*}
   \mL\mR^H-\mM_\star=\bPhi+\bPsi,~
\vDeltal\mR^H+\mL\vDeltar^H=\bPhi+2\bPsi,\numberthis  \label{eq:decomp}
    \end{align*}
where $\bPhi\coloneqq\vDeltal\mR_{\star}^H+\mL_\star\vDeltar^H$ and $\bPsi\coloneqq\vDeltal\vDeltar^H$.
\subsection{Proof of Lemma~\ref{lem:pl-fro}}\label{apd:pl-fro}For $\mF\in\B_{\mathrm{v}}(\frac{1}{64},\frac{1}{20}),$ we have $\mF\in\B_{\mathrm{s}}(\frac{1}{64})$ where $\B_{\mathrm{s}}(\cdot)$ is defined in \eqref{eq:local_region_scal}. Then 
\begin{align*}
    \norm{\nabla f_0(\mF)}_F^2=&\|(\mL\mR^H-\mM_\star)\mR(\mR^H\mR)^{-\frac{1}{2}}(\mR^H\mR)^{\frac{1}{2}}\|_F^2\\
    &+\|(\mL\mR^H-\mM_\star)^H\mL(\mL^H\mL)^{-\frac{1}{2}}(\mL^H\mL)^{\frac{1}{2}}\|_F^2
    \\\geq&(\sigma_r^2(\mL)\wedge\sigma_r^2(\mR)) (\norm{\nabla f_0(\mF)}_{\mL,\mR}^*)^2
    \\\geq&
    \frac{\sigma_r(\mM_\star)}{7}(f_0(\mF)-f_0^\star),
\end{align*}
 where in the last inequality, we invoke Lemma~\ref{lem:pl-scale} when $\mF\in\B_{\mathrm{s}}(\frac{1}{64})$  and Lemma~\ref{lem:wellcond} when $\mF\in\B_{\mathrm{v}}(\frac{1}{64},\frac{1}{20})$  .  

\subsection{Proof of Lemma~\ref{lem:pl-scale}}\label{apd:pl-scale} Let $\mE{=}\mL\mR^H-\mM_\star{=}\mM-\mM_\star$. 
  The LHS of Lemma~\ref{lem:pl-scale} is
    \begin{align*}
        \norm{\nabla f_0(\mF)}_{\mL,\mR}^*{=}\sqrt{\|\mE\mR(\mR^H\mR)^{-\frac{1}{2}}\|_F^2{+}\|\mE^H\mL(\mL^H\mL)^{-\frac{1}{2}}\|_F^2}, \numberthis \label{eq:gradscal}
    \end{align*} 
{and the RHS of Lemma~\ref{lem:pl-scale} is $\|\mL\mR^H-\mM_\star\|_F^2=\|\mM-\mM_\star\|_F^2$ (omitting the coefficient).

First verify the construction of $\Tilde{\mL}, \Tilde{\mR}, \mQ^\natural$ from $\mM$, $\mM_\star$ in \ref{addNotations}. Without loss generality, we suppose $\mL=\Tilde{\mL}\mQ^\natural$, $\mR=\Tilde{\mR}\mQ^\natural$  in what follows  as the LHS and RHS of Lemma~\ref{lem:pl-scale} remain invariant\footnote{{First verify $\mE$ is invariant as $\mL\mR^H=\mM$ from constructions.} Then we show $\|\mE\mR(\mR^H\mR)^{-\frac{1}{2}}\|_F$ is invariant, and the same holds for $\|\mE^H\mL(\mL^H\mL)^{-\frac{1}{2}}\|_F$. Verify that $\|\mE\mR(\mR^H\mR)^{-\frac{1}{2}}\|_F=\|\mE\mR(\mR^H\mR)^{-1}\mR^H\|_F$. When 
  ${\mL}\leftarrow\Tilde{\mL}\mQ^\natural$ and ${\mR}\leftarrow\Tilde{\mR}\mQ^\natural$,
  this quantity remains invariant as $\mR(\mR^H\mR)^{-1}\mR^H$ {always denotes the projection matrix to the row space of $\mM$, before and after substitution.}}.}
Denoting $\mA_L=\mE\mR(\mR^H\mR)^{-\frac{1}{2}}, \mA_R=\mE^H\mL(\mL^H\mL)^{-\frac{1}{2}}$, rewrite \eqref{eq:gradscal} as 
{
\begin{align*}
        &\norm{\nabla f_0(\mF)}_{\mL,\mR}^*=\max_{\norm{\small{\begin{bmatrix}
            \mV_L^H~ \mV_R^H\end{bmatrix}^H}}_F=1}\big|\langle\mA_L,\mV_L\rangle+\langle{\mV}_R,\mA_R\rangle\big|
.\numberthis \label{eq:dual_local_max}
\end{align*}
}
To lower bound \eqref{eq:dual_local_max}, choose
\begin{align*}
    \mV_L&={\vDeltal(\mR^H\mR)^\frac{1}{2}}/{\big\|\begin{bmatrix}
            \vDeltal^H~ \vDeltar^H\end{bmatrix}^H\big\|_{\mL,\mR},}
            \\{\mV}_R&={\vDeltar(\mL^H\mL)^\frac{1}{2}}/{\big\|\begin{bmatrix}
            \vDeltal^H~ \vDeltar^H\end{bmatrix}^H\big\|_{\mL,\mR},}
\end{align*}
where $\vDeltal=\mL-\mL_\star$ and $\vDeltar=\mR-\mR_\star$. It is easy to check that $\big\|{\begin{bmatrix}
            \mV_L^H~\mV_R^H\end{bmatrix}^H}\big\|_F=1$. 
Inserting this choice 
into \eqref{eq:dual_local_max} to obtain
            \begin{align*}
            \norm{\nabla f_0(\mF)}_{\mL,\mR}^*\geq\frac{|\langle\mE,\vDeltal\mR^H+\mL\vDeltar^H\rangle|}{\big\|\begin{bmatrix}
            \vDeltal^H~ \vDeltar^H\end{bmatrix}^H\big\|_{\mL,\mR}
}.\numberthis \label{eq:gradscal_lowerbd}
            \end{align*}
{Now we bound the numerator and the denominator in \eqref{eq:gradscal_lowerbd} separately.} Invoke the decomposition  
 \eqref{eq:decomp} to obtain:
    \begin{align*}    &\Real\langle\mE,\vDeltal\mR^H+\mL\vDeltar^H\rangle=\Real\langle\bPhi+\bPsi,\bPhi+2\bPsi\rangle
    \\&=\|\bPhi\|_F^2+2\|\bPsi\|_F^2+3\Real\langle\bPhi,\bPsi\rangle
    \\&\overset{{(a)}}{\geq}\frac{3}{4}\|\bPhi\|_F^2-7\|\bPsi\|_F^2
    \overset{{(b)}}{\geq}\frac{1}{2}\|\mL\mR^H-\mM_\star\|_F^2,
    \end{align*}
    where {in (a) we invoke $3\Real\langle \va,\vb\rangle=\|\frac{1}{2}\va+3\vb\|_2^2-\frac{1}{4}\|\va\|_2^2-9\|\vb\|_2^2\geq-\frac{1}{4}\|\va\|_2^2-9\|\vb\|_2^2$} and {(b) results from  Lemma~\ref{lem:basic-local-facts} by setting $\varepsilon=1/32$ in it.} 
    Consequently, 
\begin{align*}
    |\langle\mE,\vDeltal\mR^H+\mL\vDeltar^H\rangle|&\geq\frac{1}{2}\|\mL\mR^H-\mM_\star\|_F^2.
\end{align*}
From Lemma \ref{lem:basic-local-facts}
~one can also obtain
\begin{align*}
  \big\|\begin{bmatrix}
            \vDeltal^H~ \vDeltar^H\end{bmatrix}^H\big\|_{\mL,\mR}\leq\frac{5}{3}\|\mL\mR^H-\mM_\star\|_F.
\end{align*}
Combine the above pieces into \eqref{eq:gradscal_lowerbd} to yield
 \begin{align*}
           (\norm{\nabla f_0(\mF)}_{\mL,\mR}^*)^2\geq \frac{9}{50}(f_0(\mF)-f_0^\star).
            \end{align*}

\section{Proof of Smoothness conditions}\label{apd:proof-smooth}
Set $\mV=\mF_{+}-\mF=\begin{bmatrix}
    \mV_L^H~\mV_R^H
\end{bmatrix}^H$ where $\mV_L=\mL_{+}-\mL
$ and $\mV_R=\mR_{+}-\mR
$.
We equivalently reformulate $f_0(\mF_{+})$ as:
\begin{align*}
    &f_0(\mF_{+})\\
    =&\frac{1}{2}\|\mL\mR^H-\mM_\star\|_F^2+\Real\langle\mL_{+}\mR_{+}^H-\mL\mR^H,\mL\mR^H-\mM_\star\rangle\\&+\frac{1}{2}\|\mL_{+}\mR_{+}^H-\mL\mR^H\|_F^2\\=& f_0(\mF)+\Real\(\ip{\nabla f_0(\mF)}{\mV}+\langle{\mL\mR^H-\mM_\star},{\mV_L\mV_R^H}\rangle\)\\
    &+\frac{1}{2}\|\mL\mV_R^H+\mV_L\mR^H+\mV_L\mV_R^H\|_F^2,\numberthis\label{eq:direct-decom-smoothness}
\end{align*}
where we point out that $\mL_{+}\mR_{+}^H-\mL\mR^H=\mL\mV_R^H+\mV_L\mR^H+\mV_L\mV_R^H$ in the first equality.
\subsection{Proof of Lemma~\ref{lem:smoothness-fro}}\label{apd:smooth-fro}

We prove the lemma based on the result of \eqref{eq:direct-decom-smoothness} where  $\mV=\mF_{+}-\mF$. We first list the following facts:
$
\|\mL\mV_R^H\|_F\leq\|\mL\|\|\mV_R\|_F,~\|\mV_L\mR^H\|_F\leq\|\mR\|\|\mV_L\|_F,~
\|\mV_L\mV_R^H\|_F\leq\frac{1}{2}\|\mV\|_{F}^2.
$ Additionally invoke the basic inequality $(a+b+c)^2\leq 3a^2+3b^2+3c^2$ to obtain 
\begin{small}
$$
    \|\mL\mV_R^H+\mV_L\mR^H+\mV_L\mV_R^H\|_F^2\leq \big(3(\|\mL\|^2\vee\|\mR\|^2)
    +\frac{3}{4}\|\mV\|_F^2\big)\|\mV\|_F^2. 
$$
\end{small}Then we establish the following formulation from \eqref{eq:direct-decom-smoothness}:
\begin{align*}
  &f_0(\mF_{+})\leq f_0(\mF)+\Real\ip{\nabla f(\mF)}{\mV}+\frac{\|\mV\|_{F}^2}{2}\Big(\frac{3}{4}\|\mV\|_F^2\\
  &+3(\|\mL\|^2\vee\|\mR\|^2)+\|\mL\mR^H-\mM_\star\|_F\Big).\numberthis \label{eq:smothfro-interresult}
\end{align*}
Besides, derive the upper bound for $\mV=\mF_{+}-\mF=-\eta\nabla f(\mF)$ from the conditions exposed:
\begin{align*}
&\|\mV\|_F^2 =\eta^2\|\nabla f_0(\mF)+\nabla f_1(\mF)\|_F^2
\\&\leq\eta^2(2\|\nabla f_0(\mF)\|_F^2+2\|\nabla f_1(\mF)\|_F^2)
\\&\leq 2\eta^2(1+\varepsilon^2)(\|(\mL\mR^H-\mM_\star)\mR\|_F^2+\|(\mL\mR^H-\mM_\star)^H\mL\|_F^2)
\\&\leq5(1+\varepsilon^2)\frac{\|\mL\mR^H-\mM_\star\|_F^2}{\sigma_1(\mM_\star)}\leq\frac{1+\varepsilon^2}{80}\sigma_1(\mM_\star),
\end{align*}
where we invoke $\eta\leq1/\sigma_1(\mM_\star)$ and $\|\mL\|^2\vee\|\mR\|^2\leq\frac{5\sigma_1(\mM_\star)}{4}$ in the third inequality and the last inequality follows from $\|\mL\mR^H-\mM_\star\|_F\leq\frac{1}{20}\sigma_r(\mM_\star)\leq\frac{1}{20}\sigma_1(\mM_\star)$. Additionally inserting the results about $\|\mL\|^2\vee\|\mR\|^2$ and $\|\mL\mR^H-\mM_\star\|_F$ again into \eqref{eq:smothfro-interresult}, we prove Lemma~\ref{lem:smoothness-fro}. 
\subsection{Proof of Lemma~\ref{lem:smoothness-scale}}\label{apd:smooth-scale}
This lemma is proved based on the result of \eqref{eq:direct-decom-smoothness} where  $\mV=\mF_{+}-\mF$. We first list the following facts:
\begin{align*}
&\|\mL\mV_R^H\|_F{\leq}\|\mL(\mL^H\mL)^{-\frac{1}{2}}\|\|\mV_R(\mL^H\mL)^{\frac{1}{2}}\|_F{\leq}\|\mV_R(\mL^H\mL)^{\frac{1}{2}}\|_F,\\
&\|\mR\mV_L^H\|_F{\leq}\|\mR(\mR^H\mR)^{-\frac{1}{2}}\|\|\mV_L(\mR^H\mR)^{\frac{1}{2}}\|_F{\leq}\|\mV_L(\mR^H\mR)^{\frac{1}{2}}\|_F,\\
&\|\mV_L\mV_R^H\|_F{\leq}\frac{1}{2}\|\mV\|_{\mL,\mR}^2\|(\mR^H\mR)^{-\frac{1}{2}}(\mL^H\mL)^{-\frac{1}{2}}\|{\leq}\frac{\|\mV\|_{\mL,\mR}^2}{2\sigma_r(\mL\mR^H)},
\end{align*}
where in the last line we invoke the fact $\|\mV_L\mV_R^H\|_F=\|\mV_L(\mR^H\mR)^{\frac{1}{2}}(\mR^H\mR)^{-\frac{1}{2}}(\mL^H\mL)^{-\frac{1}{2}}(\mL^H\mL)^{\frac{1}{2}}\mV_R^H\|_F$ and the following result: 
\begin{align*}
&\|(\mR^H\mR)^{-\frac{1}{2}}(\mL^H\mL)^{-\frac{1}{2}}\|=\|\mR(\mR^H\mR)^{-1}(\mL^H\mL)^{-1}\mL^H\|\\
&=\|(\mR^H)^{\dagger}\mL^\dagger\|=\|(\mL\mR^H)^{\dagger}\|=\frac{1}{\sigma_r(\mL\mR^H)}.
\end{align*}
Additionally invoking the basic inequality $(a+b+c)^2\leq 3a^2+3b^2+3c^2$, we  establish the following result from \eqref{eq:direct-decom-smoothness}
\begin{align*}
  &f_0(\mF_{+})\leq f_0(\mF)+\Real\ip{\nabla f_0(\mF)}{\mV}\\
  &+\frac{\|\mV\|_{\mL,\mR}^2}{2}\Big(3+\frac{\|\mL\mR^H-\mM_\star\|_F}{\sigma_r(\mL\mR^H)}+\frac{3\|\mV\|_{\mL,\mR}^2}{4\sigma_r^2(\mL\mR^H)}\Big).\numberthis \label{eq:smothscale-interresult}
\end{align*}
Now we are supposed to derive the upper bound for $\|\mV\|_{\mL,\mR}^2=\|\mF_{+}-\mF\|_{\mL,\mR}^2=\eta^2\begin{bmatrix}
    \nabla_{\mL} f(\mF)(\mR^H\mR)^{-\frac{1}{2}}\\
     \nabla_{\mR} f(\mF)(\mL^H\mL)^{-\frac{1}{2}}
\end{bmatrix}=\eta^2(\|\nabla f(\mF)\|_{\mL,\mR}^{*})^2$ from the conditions exposed:
\begin{align*}
&\|\mV\|_{\mL,\mR}^2 =\eta^2(\|\nabla f_0(\mF)+\nabla f_1(\mF)\|_{\mL,\mR}^{*})^2
\\&\leq2\eta^2\((\|\nabla f_0(\mF)\|_{\mL,\mR}^{*})^2+(\|\nabla f_1(\mF)\|_{\mL,\mR}^{*})^2\)
\\&\leq 2\eta^2(1+\varepsilon^2)(\|(\mL\mR^H-\mM_\star)\mR(\mR^H\mR)^{-\frac{1}{2}}\|_F^2
\\&\quad+\|(\mL\mR^H-\mM_\star)^H\mL(\mL^H\mL)^{-\frac{1}{2}}\|_F^2)
\\&\leq 4(1+\varepsilon^2)\|\mL\mR^H-\mM_\star\|_F^2,
\end{align*}
where we invoke $\|\mL(\mL^H\mL)^{-\frac{1}{2}}\|\leq 1$, $\|\mR(\mR^H\mR)^{-\frac{1}{2}}\|\leq1$ and $\eta\leq1$. 

By the fact $\|\mL\mR^H-\mM_\star\|_F\leq\frac{1}{11}\sigma_r(\mM_\star)$ and Weyl's theorem, we have
$\sigma_r(\mL\mR^H)\geq\frac{10}{11}\sigma_r(\mM_\star)$. Consequently, 
{
$$\frac{3\|\mV\|_{\mL,\mR}^2}{4\sigma_r^2(\mL\mR^H)}\leq3(1+\varepsilon^2)\frac{\|\mL\mR^H-\mM_\star\|_F^2}{\sigma_r^2(\mL\mR^H)}\leq\frac{1}{25}+\frac{\varepsilon^2}{25}.$$
}
Finally, we prove Lemma~\ref{lem:smoothness-scale} by inserting these results into \eqref{eq:smothscale-interresult}.

\section{Upper bound of the perturbation gradient}\label{apd:E-upbd-grad}
\subsection{Proof of Lemma~\ref{lem:upbd-perturbgd-fro}}\label{apd:upbd-grad-fro}
From $\mF\in\B_{\mathrm{v}}(\frac{\delta_0}{20\sqrt{\kappa \mu_0 s}},\frac{1}{20})$ where $\delta_0\leq\frac{1}{4}$, we conclude that $\mF\in\B_{\mathrm{v}}(\frac{1}{20},\frac{1}{20})$ and $\mF\in\B_{\mathrm{s}}(\frac{\delta_0}{20\sqrt{\kappa\mu_0 s}})$ where $\B_{\mathrm{s}}(\cdot)$ is defined in \eqref{eq:local_region_scal}.
Set $\mE=\G(\A^{*}\A-\I)\G^{*}(\mL\mR^H-\mM_\star)$, 
\begin{align*}
    \|\nabla f_1(\mF)\|_F^2&=\|\mE\mR(\mR^H\mR)^{-\frac{1}{2}}(\mR^H\mR)^\frac{1}{2}\|_F^2\\
    &\quad+\|\mE^H\mL(\mL^H\mL)^{-\frac{1}{2}}(\mL^H\mL)^\frac{1}{2}\|_F^2\\
    &\leq (\sigma_1^2(\mL)\vee\sigma_1^2(\mR))(\|\nabla f_1(\mF)\|_{\mL,\mR}^{*})^2
    \\&\leq\frac{\sqrt{5}}{2}\sigma_1(\mM_\star)^{\frac{1}{2}}(\|\nabla f_1(\mF)\|_{\mL,\mR}^{*})^2,
\end{align*}
where we invoke Lemma~\ref{lem:wellcond} in the last inequality. 

 The upper bound for $\|\nabla f_1(\mF)\|_{\mL,\mR}^{*}$ is established in Lemma~\ref{lem:upbd-gd-scale} whose proof is in Appendix \ref{apd:upbd-grad-scale}. We
apply an intermediary result \eqref{eq:upbdfv-gd-scale} of Appendix \ref{apd:upbd-grad-scale} to finish our proof:
$$\|\nabla f_1(\mF)\|_{\mL,\mR}^{*}\leq\frac{3\delta}{2}\sqrt{2f_0(\mF)},$$when $\mF\in\B_{\mathrm{s}}(\frac{\delta}{20\sqrt{\mu_0 s}})$ and $n\gtrsim O(\delta \mu_0\mu_1sr\log(sn))$. 
{By setting $\delta=\delta_0 /\sqrt{\kappa}$, then for $\mF\in\B_{\mathrm{s}}(\frac{\delta_0}{20\sqrt{\kappa\mu_0 s}})$:} 
\begin{align*}
    \|\nabla f_1(\mF)\|_F{\leq} \frac{3\sqrt{10}\delta_0\sigma_1(\mM_\star)^{\frac{1}{2}}}{4\sqrt{\kappa}}\sqrt{f_0(\mF)}{\leq}\frac{25}{4}\delta_0\|\nabla f_0(\mF)\|_F,
\end{align*}
where we invoke Lemma~\ref{lem:pl-fro} in the second inequality. 
\subsection{Proof of Lemma~\ref{lem:upbd-gd-scale}}\label{apd:upbd-grad-scale}
Set $\mE=\G(\A^{*}\A-\I)\G^{*}(\mL\mR^H-\mM_\star)$, we know
\begin{align*}
\|\nabla f_1(\mF)\|_{\mL,\mR}^{*}=&\sqrt{\|\mE\mR(\mR^H\mR)^{-\frac{1}{2}}\|_F^2+\|\mE^H\mL(\mL^H\mL)^{-\frac{1}{2}}\|_F^2}
\\=&\sqrt{I_1^2+I_2^2}\leq I_1+I_2,
\end{align*}
where $I_1=\|\mE\mR(\mR^H\mR)^{-\frac{1}{2}}\|_F$, $I_2=\|\mE^H\mL(\mL^H\mL)^{-\frac{1}{2}}\|_F$ and  we will bound $I_1$ and $I_2$ separately. 

{As in Appendix 
\ref{apd:pl-scale}, we suppose $\mL=\Tilde{\mL}\mQ^\natural$ and $\mR=\Tilde{\mR}\mQ^\natural$ as the quantity $I_1$, $I_2$ remain invariant.}  $\Tilde{\mL}$, $\Tilde{\mR}$ and $\mQ^\natural$ are defined in Appendix \ref{addNotations}. 
As $I_1=\max_{\|\Tilde{\mV}\|_F=1} |\langle\mE\mR(\mR^H\mR)^{-\frac{1}{2}},\Tilde{\mV}\rangle|$, 
we point out the following splitting
\begin{align*}
&|\langle\mE\mR(\mR^H\mR)^{-\frac{1}{2}},\Tilde{\mV}\rangle|=|\langle\mE,\Tilde{\mV}(\mR^H\mR)^{-\frac{1}{2}}\mR^H\rangle|\\
    &\leq|\langle\G(\A^{*}\A-\I)\G^{*}(\bPhi+\bPsi),\Tilde{\mV}(\mR^H\mR)^{-\frac{1}{2}}(\mR_\star+\vDeltar)^H\rangle|
    \\&\leq|\langle\P_T\G(\A^{*}\A-\I)\G^{*}\P_T(\bPhi),\Tilde{\mV}(\mR^H\mR)^{-\frac{1}{2}}\mR_\star^H\rangle|
    \\&\quad+|\langle\P_T\G(\A^{*}\A-\I)\G^{*}(\bPsi),\Tilde{\mV}(\mR^H\mR)^{-\frac{1}{2}}\mR_\star^H\rangle|
    \\&\quad+|\langle\G(\A^{*}\A-\I)\G^{*}\P_T(\bPhi),\Tilde{\mV}(\mR^H\mR)^{-\frac{1}{2}}\vDeltar^H\rangle|
    \\&\quad+|\langle\G(\A^{*}\A-\I)\G^{*}(\bPsi),\Tilde{\mV}(\mR^H\mR)^{-\frac{1}{2}}\vDeltar^H\rangle|
    \\&=J_1+J_2+J_3+J_4,
\end{align*}
where {$T$ is defined in Lemma~\ref{lem:tancon_inq}}, 
and $\vDeltar$, $\bPhi$, $\bPsi$ are defined in Appendix \ref{addNotations}.

 In what follows, we apply Lemma~\ref{lem:basic-local-facts} many times, setting $\varepsilon=\delta/10\sqrt{\mu_0 s}$ in it.
We first point out two results that are useful for bounding the previous quantities:
\begin{align*}
    \|\Tilde{\mV}(\mR^H\mR)^{-\frac{1}{2}}\mR_\star^H\|_F&\leq\|\Tilde{\mV}\|_F\|(\mR^H\mR)^{-\frac{1}{2}}\vSS_\star^{\frac{1}{2}}\|
    \leq\frac{1}{1-\delta},
    \\
    \|\Tilde{\mV}(\mR^H\mR)^{-\frac{1}{2}}\vDeltar^H\|_F&\leq\|\Tilde{\mV}\|_F\|(\mR^H\mR)^{-\frac{1}{2}}\vDeltar^H\|\leq\frac{\delta/(1-\delta)}{10\sqrt{\mu_0 s}},
\end{align*}
where we use the facts $\|\Tilde{\mV}\|_F=1$, $\mu_0 s\geq 1$, \eqref{eq:diag-div-scale-bd} and \eqref{eq:dd-scale-upbd} in Lemma~\ref{lem:basic-local-facts}. Then we are prepared to derive the detailed bounds, considering the fact $\delta\leq1/4$:
\begin{align*}
    J_1&\leq \|\P_T\G(\A^{*}\A-\I)\G^{*}\P_T\|\|\bPhi\|_F\|\Tilde{\mV}(\mR^H\mR)^{-\frac{1}{2}}\mR_\star^H\|_F\\
    &\leq\frac{\delta(1+\delta/10)}{10(1-\delta)}\|\mL\mR^H-\mM_\star\|_F\leq \frac{7\delta}{50}\|\mL\mR^H-\mM_\star\|_F,
    \end{align*}
    where we invoke Lemma~\ref{lem:tancon_inq} 
    and Lemma~\ref{lem:basic-local-facts} for  $\|\bPhi\|_F$.
\begin{align*}
        J_2&{\leq} \|\P_T\G(\A^{*}\A-\I)\G^{*}\|\|\bPsi\|_F\|\Tilde{\mV}(\mR^H\mR)^{-\frac{1}{2}}\mR_\star^H\|_F\\
    &{\leq}\frac{\delta(\sqrt{\mu_0 s(1{+}\frac{\delta}{10})}{+}1)}{10\sqrt{\mu_0 s}(1-\delta)}\|\mL\mR^H{-}\mM_\star\|_F
    {\leq} \frac{3\delta}{10}\|\mL\mR^H{-}\mM_\star\|_F,
    \end{align*}
    where we know  $\|\P_T\G(\A^{*}\A-\I)\G^{*}\|\leq\|\P_T\G\A^{*}\|\|\A\G^{*}\|+\|\P_T\G\G^{*}\|\leq \sqrt{\mu_0 s(1+\delta/10)}+1$, by $\|\P_T\|,\|\G\|\leq 1$, $\|\A\|\leq\sqrt{\mu_0 s}$ in \cite{Chen2022} and $\|\P_T\G\A^{*}\|=\|\A\G^{*}\P_T\|\leq\sqrt{1+\delta/10}$ from Lemma~\ref{lem:semi-tancon}. Besides, we invoke $\|\bPsi\|_F\leq\delta/10\sqrt{\mu_0 s}$ from Lemma~\ref{lem:basic-local-facts}. 
    
    Similarly from $\|\G(\A^{*}\A-\I)\G^{*}\P_T\|\leq\sqrt{\mu_0 s(1+\delta/10)}+1$, we can establish the bound for $J_3$
\begin{align*}
        J_3&\leq\|\G(\A^{*}\A-\I)\G^{*}\P_T\|\|\bPhi\|_F\|\Tilde{\mV}(\mR^H\mR)^{-\frac{1}{2}}\vDeltar^H\|_F
    \\&\leq \frac{3\delta}{10}\|\mL\mR^H-\mM_\star\|_F.
\end{align*}
Invoking $\|\bPsi\|_F\leq\delta/10\sqrt{\mu_0 s}$ from Lemma~\ref{lem:basic-local-facts} and the fact $\|\G(\A^{*}\A-\I)\G^{*}\|\leq\|\G\A^{*}\A\G^{*}\|+\|\G\G^{*}\|\leq 2\mu_0 s$, we have
\begin{align*}
     J_4&\leq\|\G(\A^{*}\A-\I)\G^{*}\|\|\bPsi\|_F\|\Tilde{\mV}(\mR^H\mR)^{-\frac{1}{2}}\vDeltar^H\|_F
     \\&
     \leq \frac{\delta}{100}\|\mL\mR^H-\mM_\star\|_F.
\end{align*}
Then we obtain $I_1\leq \frac{3\delta}{4}\|\mL\mR^H-\mM_\star\|_F$, and the same holds for $I_2$ that $I_2\leq \frac{3\delta}{4}\|\mL\mR^H-\mM_\star\|_F$. Then we establish 
\begin{align*}
    \|\nabla f_1(\mF)\|_{\mL,\mR}^{*}\leq 3\delta\sqrt{2f_0(\mF)}/2,
    \numberthis\label{eq:upbdfv-gd-scale}
\end{align*}
and combine $\sqrt{2f_0(\mF)}\leq \frac{10\|\nabla f_0(\mF)\|_{\mL,\mR}^{*}}{3}$ from Lemma~\ref{lem:pl-scale}, we finish the proof.

\section{Technical Lemmas}\label{apd:F-support-lems}
\begin{lemma}[\cite{Chen2022}{, Corollary \uppercase\expandafter{\romannumeral3}.9}]
	\label{lem:tancon_inq}
 Let $T$ be the tangent space of $\mM_\star$, which is {defined} as $T=\{\mU_\star\mC^H+\mD\mV_\star^H|\mC\in\C^{n_2\times r},\mD\in\C^{sn_1\times r}\}$. 	Under Assumption \ref{assumption1}, Assumption \ref{assumption2} and $n\geq C\varepsilon^{-2}\mu_0\mu_1 sr\log(sn)$, we have
	\begin{align*}
		\|{\P_{T}\G(\A^{*}\A - \I)\G^{*}\P_{T}}\|\leq \varepsilon \numberthis \label{eq:event-tancon-inq}
	\end{align*}
 with probability at least $1-(sn)^{-c}$ for 
  constants $c,C>0$.
\end{lemma}    
\begin{lemma}\label{lem:semi-tancon}
Under event \eqref{eq:event-tancon-inq}, one has
$
    \|\A\G^{*}\P_T\|\leq \sqrt{1+\varepsilon}.
$
 \begin{proof}
Rewrite that
\begin{align*}
&\|\A\G^{*}\P_T(\mM)\|_F^2=\langle\P_T\G\A^{*}\A\G^{*}\P_T(\mM),\mM\rangle\\
&=\langle\P_T\G(\A^{*}\A-\I)\G^{*}\P_T(\mM),\mM\rangle+\langle\P_T\G\G^{*}\P_T(\mM),\mM\rangle
\\&\leq(1+\varepsilon)\|\mM\|_F^2,
\end{align*}where we use the condition \eqref{eq:event-tancon-inq}. 
 \end{proof}   
\end{lemma}
 \begin{lemma}[Well-conditionedness of factors]\label{lem:wellcond}
 For  $\mF\in\B_{\mathrm{v}}(\frac{1}{20},\frac{1}{20})$, one has {well-conditionedness} of factors $\mL$, $\mR:$
 \begin{align*}
   \frac{4\sigma_r(\mM_\star)}{5}&\leq \sigma_r^2(\mL)\wedge\sigma_r^2(\mR)\leq\sigma_1^2(\mL)\vee\sigma_1^2(\mR)\leq\frac{5\sigma_1(\mM_\star)}{4}.
 \end{align*}
 \end{lemma}
 \begin{proof}
   From Lemma~B.4 in \cite{Zilber2022}
   , we have 
   \begin{align*}
       &\|\mL\hat{\mQ}-\mL_\star\|_F^2+\|\mR\hat{\mQ}-\mR_\star\|_F^2\\&\leq\frac{\sqrt{2}+1}{\sigma_r(\mM_\star)}\Big(\|\mL\mR^H-\mM_\star\|_F^2{+}\frac{1}{4}\|\mL^H\mL-\mR^H\mR\|_F^2\Big),
   \end{align*}
   where $\hat{\mQ}=\mathrm{sgn}(\mF^H\mF_\star)=\vU\vV^H$ if the SVD of $\mF^H\mF_\star$ is $\vU\vSS\vV^H$. Thus 
   $\|\mL\hat{\mQ}-\mL_\star\|_F^2\leq 0.01\sigma_r(\mM_\star)$ 
   for $\mF\in\B_{\mathrm{v}}(\frac{1}{20},\frac{1}{20})$. 
    From Weyl's theorem, we obtain that 
    \begin{align*}
     &\sigma_1(\mL){=}\sigma_1(\mL\hat{\mQ}){\leq} \|\mL\hat{\mQ}-\mL_\star\|_F{+}\sigma_1(\mL_\star){\leq}\sqrt{5}\sigma_1(\mM_\star)^\frac{1}{2}/2,
     \\
&\sigma_r(\mL){=}\sigma_r(\mL\hat{\mQ}){\geq}\sigma_r(\mL_\star){-}\|\mL\hat{\mQ}-\mL_\star\|_F{\geq} 2\sigma_r(\mM_\star)^\frac{1}{2}/\sqrt{5},
    \end{align*}
   and the same bounds hold for $\sigma_1(\mR)$ and $\sigma_r(\mR)$. 
 \end{proof}
{In the following lemmas, $\mL$, $\mR$, 
$\vDeltal$, $\vDeltar$, $\bPhi$ and $\bPsi$  are defined in additional notations in Appendix \ref{addNotations}.}
\begin{lemma} \label{lem:distbl_upbd}
One has 
\begin{align*}
   \|\vDeltal\vSS_{\star}^{\frac{1}{2}}\|_F^2+\|\vDeltar\vSS_{\star}^{\frac{1}{2}}\|_F^2\leq (\sqrt{2}+1)\|\mL\mR^H-\mM_\star\|_F^2.
\end{align*}
\begin{proof}
   This is an intermediary result in the proof of Lemma~24 in \cite{Tong2021} and Lemma~5.4 in \cite{Tu2016}. 
\end{proof}
\end{lemma}

\begin{lemma}\label{lem:basic-local-facts}
When $\|\mL\mR^H-\mM_\star\|_F\leq 0.5\varepsilon\sigma_r(\mM_\star)$, one has  \begin{align*}   
     \|\bPsi\|_F &\leq\varepsilon\|\mL\mR^H-\mM_\star\|_F,
    \numberthis\label{eq:psi_upbd}
    \\
      (1-\varepsilon)\|\mL\mR^H-\mM_\star\|_F &\leq\|\bPhi\|_F\leq(1+\varepsilon)\|\mL\mR^H-\mM_\star\|_F,\numberthis\label{eq:phi_bd}
\\
    \|(\mL^H\mL)^{-\frac{1}{2}}\vSS_\star^{\frac{1}{2}} \|\vee&\|(\mR^H\mR)^{-\frac{1}{2}}\vSS_\star^{\frac{1}{2}}\|\leq \frac{1}{1-\varepsilon},\numberthis \label{eq:diag-div-scale-bd}
 \\   \|\vDeltal(\mL^H\mL)^{-\frac{1}{2}}\|\vee&\|\vDeltar(\mR^H\mR)^{-\frac{1}{2}}\|\leq \frac{\varepsilon}{1-\varepsilon},  \numberthis \label{eq:dd-scale-upbd}    
   \\\big\|{\begin{bmatrix}
            \vDeltal^H~ \vDeltar^H\end{bmatrix}^H}\big\|_{\mL,\mR}\leq&\frac{8}{5}(1+\varepsilon)\|\mL\mR^H-\mM_\star\|_F.
\numberthis\label{eq:local_norm_ddupbd}
\end{align*}
\end{lemma}
\begin{proof}
From the relation that $\|\mA\mB\|_F\geq\sigma_r(\mB)\|\mA\|_F$ and $\|\mA\|_F\geq\|\mA\|$, one obtains 
\begin{align*}
    \|\vDeltal\vSS_\star^{-\frac{1}{2}}\|\vee\|\vDeltar\vSS_\star^{-\frac{1}{2}}\|\leq\frac{\|\vDeltal\vSS_\star^{\frac{1}{2}}\|_F\vee\|\vDeltar\vSS_\star^{\frac{1}{2}}\|_F}{\sigma_r(\mM_\star)}.
\end{align*}
By Lemma~\ref{lem:distbl_upbd} and $\|\mL\mR^H-\mM_\star\|_F\leq 0.5\varepsilon\sigma_r(\mM_\star)$ to obtain
\begin{align*}
    \|\vDeltal\vSS_\star^{-\frac{1}{2}}\|\vee\|\vDeltar\vSS_\star^{-\frac{1}{2}}\|
    &\leq \sqrt{2}(\sqrt{2}+1)^{-\frac{1}{2}}\varepsilon. \numberthis\label{eq:useful-fact}
\end{align*}
Setting $p=\|\vDeltal\vSS_\star^{-\frac{1}{2}}\|\vee\|\vDeltar\vSS_\star^{-\frac{1}{2}}\|$ and 
from the facts that $\|\vDeltal\vDeltar^H\|_F\leq\|\vDeltal\vSS_\star^{-\frac{1}{2}}\|\|\vDeltal\vSS_\star^{\frac{1}{2}}\|_F$ and $\|\vDeltal\vDeltar^H\|_F\leq\|\vDeltar\vSS_\star^{-\frac{1}{2}}\|\|\vDeltar\vSS_\star^{\frac{1}{2}}\|_F$, we have
\begin{align*}
    \|\bPsi\|_F=\|\vDeltal\vDeltar^H\|_F&\leq\frac{p}{2}\(\|\vDeltal\vSS_\star^{\frac{1}{2}}\|_F+\|\vDeltar\vSS_\star^{\frac{1}{2}}\|_F\)
    \\
    &\leq \varepsilon\|\mL\mR^H-\mM_\star\|_F,
\end{align*}
   where 
   in the last inequality we invoke Lemma~\ref{lem:distbl_upbd}.

    We know that $\|\bPhi\|_F=\|\mL\mR^H-\mM_\star-\bPsi\|_F$, and use the triangle inequality to obtain
  $\|\mL\mR^H-\mM_\star\|_F-\|\bPsi\|_F \leq\|\bPhi\|_F\leq\|\mL\mR^H-\mM_\star\|_F+\|\bPsi\|_F,$
    thus we establish \eqref{eq:phi_bd}.
    
    For \eqref{eq:diag-div-scale-bd}, we point out the relation that
\begin{align*}
    \|(\mL^H\mL)^{-\frac{1}{2}}\vSS_\star^{\frac{1}{2}}\|&=\|\mL(\mL^H\mL)^{-1}\vSS_{\star}^{\frac{1}{2}}\|
    \\&\leq\frac{1}{1-\|\vDeltal\vSS_{\star}^{-\frac{1}{2}}\|}\leq\frac{1}{1-\varepsilon},
\end{align*}
where we invoke Lemma~15 of \cite{Tong2021} 
in the first inequality and the fact \eqref{eq:useful-fact} in the last inequality. Similarly, we establish that $\|(\mR^H\mR)^{-\frac{1}{2}}\vSS_\star^{\frac{1}{2}}\|\leq\frac{1}{1-\varepsilon}$. For \eqref{eq:dd-scale-upbd},
    \begin{align*}
        \|\vDeltal(\mL^H\mL)^{-\frac{1}{2}}\|=\|\vDeltal\vSS_\star^{-\frac{1}{2}}\|\|(\mL^H\mL)^{-\frac{1}{2}}\vSS_\star^{\frac{1}{2}}\|\leq\frac{\varepsilon}{1-\varepsilon},
    \end{align*}
    and similarly we obtain $\|\vDeltar(\mR^H\mR)^{-\frac{1}{2}}\|\leq\frac{\varepsilon}{1-\varepsilon}$.  For \eqref{eq:local_norm_ddupbd},
    \begin{align*}
    \big\|{\begin{bmatrix}
            \vDeltal^H~ \vDeltar^H\end{bmatrix}^H}\big\|_{\mL,\mR}^{2}=\|\vDeltal(\mR^H\mR)^{\frac{1}{2}}\|_F^2+\|\vDeltar(\mL^H\mL)^{\frac{1}{2}}\|_F^2.
\end{align*}
We first bound $ \|\vDeltal(\mR^H\mR)^{\frac{1}{2}}\|_F$, and bounding $ \|\vDeltar(\mL^H\mL)^{\frac{1}{2}}\|_F$ follows from a similar route:
\begin{align*}
    \|\vDeltal(\mR^H\mR)^{\frac{1}{2}}\|_F&\leq\|\vDeltal\vSS_\star^{\frac{1}{2}}\|_F\|\vSS_\star^{-\frac{1}{2}}(\mR^H\mR)^\frac{1}{2}\|\\&=\|\vDeltal\vSS_\star^{\frac{1}{2}}\|_F\|\mR\vSS_\star^{-\frac{1}{2}}\|\leq(1+\varepsilon)\|\vDeltal\vSS_\star^{\frac{1}{2}}\|_F,
\end{align*}
where invoke the fact $\|\mR\vSS_\star^{-\frac{1}{2}}\|=\|(\mR_\star+\vDeltar)\vSS_\star^{-\frac{1}{2}}\|\leq 1+\|\vDeltar\vSS_\star^{-\frac{1}{2}}\|\leq1+\varepsilon$ in the last inequality. Together with Lemma~\ref{lem:distbl_upbd}, we obtain that 
\begin{align*}
    \big\|{\begin{bmatrix}
            \vDeltal^H~ \vDeltar^H\end{bmatrix}^H}\big\|_{\mL,\mR}\leq&\frac{8}{5}(1+\varepsilon)\|\mL\mR^H-\mM_\star\|_F.
\end{align*}

\end{proof}

\section*{Acknowledgments}
The authors would like to thank the anonymous reviewers and the Associate Editor for their constructive comments which have helped to improve the quality of this work. The authors also thank Jinchi Chen for helpful discussions.




\end{document}